\documentclass[sn-mathphys,Numbered]{sn-jnl}

\usepackage{graphicx}%
\usepackage{multirow}%
\usepackage{amsmath,amssymb,amsfonts}%
\usepackage{amsthm}%
\usepackage{mathrsfs}%
\usepackage[title]{appendix}%
\usepackage{xcolor}%
\usepackage{textcomp}%
\usepackage{manyfoot}%
\usepackage{booktabs}%
\usepackage{algorithm}%
\usepackage{algorithmicx}%
\usepackage{algpseudocode}%
\usepackage{listings}%
\usepackage{subcaption}%

\newtheorem{theorem}{Theorem}
\begin{document}

\title[Article Title]{A Dynamic Agent Based Model of the Real
Economy with Monopolistic Competition,
Perfect Product Differentiation,
Heterogeneous Agents, Increasing Returns
to Scale and Trade in Disequilibrium}


\author*[1]{\fnm{Subhamon} \sur{Supantha}}\email{subhamonsupantha@gmail.com}

\author[2]{\fnm{Naresh Kumar} \sur{Sharma}}\email{nksharma@uohyd.ac.in}

\affil*[1]{\orgdiv{Department of Computer Science}, \orgname{Chennai Mathematical Institute}, \orgaddress{\street{ H1, Sipcot IT Park, Kelambakkam, Siruseri}, \city{Chennai}, \postcode{603103}, \state{Tamil Nadu}, \country{India}}}

\affil[2]{\orgdiv{School of Economics}, \orgname{University of Hyderabad}, \orgaddress{\street{India Post Lingampally CR Rao Road}, \city{Hyderabad}, \postcode{500046}, \state{Telangana}, \country{India}}}


\abstract{We have used agent-based modeling as our numerical method to artificially simulate a dynamic real economy where agents are rational maximizers of an objective function of Cobb-Douglas type. The economy is characterised by heterogeneous agents, acting out of local or imperfect information, monopolistic competition, perfect product differentiation, allowance for increasing returns to scale technology and trade in disequilibrium. An algorithm for economic activity in each period is devised and a general purpose open source agent-based model is developed which allows for counterfactual inquiries, testing out treatments, analysing causality of various economic processes, outcomes and studying emergent properties. 10,000 simulations, with 10 firms and 80 consumers are run with varying parameters and the results show that from only a few initial conditions the economy reaches equilibrium while in most of the other cases it remains in perpetual disequilibrium. It also shows that from a few initial conditions the economy reaches a disaster where all the consumer wealth falls to zero or only a single producer remains. Furthermore, from some initial conditions, an ideal economy with high wage rate, high consumer utility and no unemployment is also reached. It was also observed that starting from an equal endowment of wealth in consumers and in producers, inequality emerged in the economy. In majority of the cases most of the firms(6-7) shut down because they were not profitable enough and only a few firms remained. Our results highlight that all these varying outcomes are possible for a decentralized market economy with rational optimizing agents.}




\maketitle

\section{Introduction}\label{sec1}

An economy is a complex system with heterogeneous agents acting for their own personal objective. Firms are maximising profit and consumers are maximising utility. Dependence of each agent’s action on that of others and feedback mechanism through price signalling give rise to complexity.

When production functions with increasing returns to scale are introduced, it ceases to be convex and the global maximum is no longer necessarily a stationary point and according to Arrow-Debreu theorem, production with such functions does not necessitate the existence of a market equilibrium. \cite{arrowdebreu} In addition, when agents don’t have global information about market demand curves, the only way to reach market clearing prices is through classical \textit{tatonnement}. However, Sonnenschein-Mantel-Debreu theorem rules out the necessity of such a process to reach an equilibrium, even if it exists. \cite{ackerman} Hence, in this model trade has been allowed for in disequilibrium. 

The aggregation of capital has been a source of controversy in Capital Theory. \cite{Joan} The approach in this model has been to treat each non-labour input as a separate parameter in the production function of a firm in order to avoid aggregation. Also, each firm is assumed to produce a unique, perfectly differentiated good. From the set of all available products, each agent consumes a subset of it and all other goods which are not in the subset, do not affect their objective function. All objective functions are taken to be of Cobb-Douglas type, however the coefficient of elasticities, the inputs and the number of inputs are all allowed to vary from agent to agent. 

With all the products being unique and perfectly differentiated, all the producers are monopolies. However as substitutability is allowed for in a cobb-douglas production function, demanders respond to price changes and hence there exists monopolistic competition. An important point to be mentioned however is that though products have been assumed to be heterogeneous, labour has been considered to be homogeneous. Also, as already mentioned, agents do not have global information and hence are acting out of local knowledge. With the dynamic nature of the model, agents have to plan a priori, for example, target an output before actually purchasing the inputs necessary to do so. In all such cases, owing to lack of perfect information, agents are working out of naive expectations. That is to say, for all variables whose value at time $t+1$ is not deterministically known to an agent at time $t$, their best guess for the value at $t+1$ is the value at $t$. \cite{Expectations}

Analytical closed form solutions for the time paths of production, consumption, prices and the various economic variables for all the agents, with such a high amount of complexity, may not exist or even if it does would be extremely tedious to find. A numerical approach with Agent-Based Modeling(ABM) has thus been taken. In the ABM approach, artificial agents are programmed to act according to the rules elaborated in the next section which gives rise to an artificial economy. The economy thus constructed is sensitively dependent on initial conditions which are externally provided by the user. The Artificial Economy(ies) thus simulated can be observed and various macroeconomic variables of interest can be studied. Counterfactuals can be tested by creating ceteris paribus conditions against any treatment and causal analysis can be performed. \cite{tesfatsion2006agent} \cite{arthur2006out}

A general-purpose ABM has been created for the real economy and 10,000 simulations with varying initial conditions were run on the model thus created. The results from these simulations show that only from some small pockets of initial conditions the model eventually reaches equilibrium. A few typical cases have been discussed in more detail and some general observations are laid down in the results section.

The modelling has been done in Python with the MESA module being used extensively. \cite{mesa}

\section{The Model}
\subsection{The Economy as a Network}
The Economy is a connected directed graph $G$ of agents as nodes and an edge $(a,b)$, $a,b \in V(G)$ means that $a$ sells to $b$. For any particular node $y \in V(G)$ it’s providers are all those producers who sell to $y$, defined formally as the set $\rho(y)= \{x \in V(G) : (x,y) \in X(G)\}$. It must be noted that $V(G)$ and $X(G)$ refer to the vertex set and edge set respectively of a graph G. The out-neighbourhood of a node $y$ is the set of all nodes $w$ such that $(y,w) \in X(G)$.  

\subsection{The General Structure}
Agents are classified into consumers and producers. Producers produce a unique differentiated good by using other such goods from its providers and labour from consumers as its input. The objective of a producer is to maximise their profit function. 

The consumers buy goods from their providers. Such goods are positively valued by the consumers and so is leisure and income for the next period. The consumers are maximising their utility function.

At time $t = 0$, all agents are endowed with equal wealth $W_0$, all producers have inventory stock of their goods exactly equal to the demand they will face at $t=1$. All initial commodity prices, wage rate, production function and utility functions are set exogenously. In addition, a matrix $S$ is exogenously given where each element $s_{ij}$ denotes the $j^{th}$ consumer’s share of $i^{th}$ firm, $0\leq s_{ij} \leq 1$ and $\sum_{j=1}^{n_c} s_{ij} = 1$, where($n_c$ is the number of consumers in the economy). At the end of each period the firm $i$ keeps a fraction of its profit and adds it to its wealth and redistributes the rest to the consumers. The fraction thus described will be called profit reinvestment ratio(PRR), $0 \leq PRR \leq 1$. In the model constructed, PRR is set initially and is the same for all producers. Hence the profit income of an agent $j$ from a firm $i$ at a time $t$ will be given as $V_{jt} = s_{ij}(1-PRR)\pi_{it}$. Apart from the variables thus mentioned, the adjustment factors of each agent and that of the economy is also set (to be explained). All these parameters which are exogenously assigned at time $t=0$, will be called initial conditions. 

At each time period the agents act as follows and in the order in which they are mentioned. 

\begin{algorithm}
\caption{Agent action procedure for each period}\label{algo1}
\begin{algorithmic}[1]

    \State The agents gather information about the prices of the products of their providers. 
    \State They calculate their demand for goods(either as input for production or for consumption) and labour(applicable only for producers), leisure and income for the next period. . 
    \State All the agents send their demands to their respective providers. 
    \State The producers sell their demanders goods from their inventory.
    \State Labour supply and labour demand of the whole demand is calculated and is bought and sold at an aggregate level, with a wage rate being the same across the economy. 
    \State After step 5, all trade for that period has been completed. The agents calculate their costs and the producers produce the goods for the next periods and augment that to the inventory. If the producer's inventory falls to zero, then that producer is marked for removal. 
    \State Each individual producer calculates the excess demand and changes their commodity price in the direction of the excess demand. 
    \State The aggregate excess labour demand is calculated and the wage rate is changed in the direction of the excess demand. 
    \State Each individual firm, if it has made profit in that period, redistributes $(1-PRR)$ amount of it among its shareholders. If the firm has been marked for removal, its PRR is set to 0. 
    \State Each consumer calculates their own utility and income earned in this period and adds it to their stock of wealth. 
    \State All demanders of a firm marked for removal have a choice to either remove the shut-down firm from their provider set or replace it with one of the available producers not already in their provider set. This decision is assumed to be stochastic and each agent has a probability of 0.5 of choosing either of the two outcomes. This choice is however only there when there exists possible candidates for the outgoing firm to be replaced by. In either case, the output elasticities of the producers are renormalized and multiplied with the degree of  homogeneity, which remains unchanged throughout.  If there is no producer left in the economy, the model program halts.
    \State If any consumer is left with no provider, their utility function parameters and provider set is regenerated.
    \State The shut-down firms are removed from the economy and start the next period. 

\end{algorithmic}
\end{algorithm}

\subsection{The Producer and the Production Function}
The production function of a producer $y$ is given as follows.
\begin{equation}
    Q_{yt} = A_{y} L_{yt}^\beta \prod_{i=1}^\eta Q_{x_it}^{\alpha_i} 
\end{equation}

Where $Q_{yt}$ is the quantity produced by the producer $y$ at time $t$. $A_y$ is the technological constant of $y$, $L_{yt}$ is the labour used in time $t$ for production. $\beta$ is the output elasticity of labour, $\eta$ is the number of providers of $y$ and $Q_{x_it}$ is the quantity of product of the provider $x_i \in \rho(y)$  at time $t$ used in the production. $\alpha_i$ is the output elasticity of the good produced by $x_i$ and $\beta$ is the output elasticity of labour. $ \forall i, 0 < \alpha_i, \beta < 1$, or in other words, all inputs have diminishing marginal product. The degree of  homogeneity is a scalar $\kappa$ such that for any homogeneous production function, $f$, $\kappa = \frac{k'}{k}$ when $ \forall k > 1, f (k \textbf{x}) = {k'} f(\textbf{x})$.

In a Cobb Douglas production function, the degree of homogeneity is the sum of all  the output elasticity coefficients. Hence in our case it is $\kappa = \sum_{i=1}^\eta \alpha_i + \beta $. A production function has increasing, decreasing or constant returns to scale if $\kappa >1 , \kappa <1  $ or $\kappa =1$ respectively. 

$Q_{x_it}$ is the minimum of the quantity supplied by the provider $x_i$ to $y$ and the quantity demanded by $y$ from $x_i$ at time $t$. If $D_t(y)$ is defined as the out-neighbourhood of $S$ at any time $t$, the market demand of any producer $y$ is the sum of all the individual demands from its out-neighbours, $d_{it} \in D_t(y)$.   If  the individual demand of an agent $d_{it}$ for the good produced by $y$ is given as $q_{d_iyt}$ and the market demand for $y$ at time $t$ is given as $Q_{yt}^d$, then 
\begin{equation}
    Q_{yt}^d = \sum_{d_{it} \in D_t(y)} q_{d_iyt} 
\end{equation}

The quantity supplied by any firm $y$ at time $t$, denoted by $Q_{yt}^s$ is given by 
\begin{equation}
    Q_{yt}^s = I_{yt} 
\end{equation}
Where $I_{yt}$ is the inventory stock of the product of $y$ at time $t$. It is updated by the relation 
\begin{equation}
    I_{y(t+1)}  =Q_{y(t+1)}^s = Q_{yt}^s -min(\{Q_{yt}^s, Q_{yt}^d\}) + Q_{yt}
\end{equation}

At every time period, for any particular producer $y$ if $Q_{yt}^d > Q_{yt}^s$, then every individual agent $d_i$ who demanded from $y$ gets $(\frac{Q_{yt}^s}{Q_{yt}^d})q_{d_iyt}$ of the good.  

The cost function of any agent $y$ is the linear combination of the quantity purchased with their respective prices as weights. 
\begin{equation}
    C_{yt} = \omega_t L_{yt} + \sum_{i=1}^{\eta}P_{x_it}q_{x_it}
\end{equation}
Where $x_i \in \rho(y)$, $\omega$ is the wage rate, $L$ is the amount of labour purchased by $y$ at $t$. 

The profit earned at period $t$ is given as 
\begin{equation}
    \pi_{yt} = R_{yt} - C_{yt}
\end{equation}
Where $R_{yt}$ is the revenue earned at period $t$.  
\begin{equation}
    R_{yt} =  P_t  min(\{Q_{yt}^s,Q_{yt}^d\}) \newline
    \Rightarrow R_{yt} =  P_t  min(\{I_{yt},Q_{yt}^d\})
\end{equation}

The decision of how much to produce and how much of inputs to demand requires some attention. It is obtained from maximising the profit function given as. 

\begin{equation}
    \pi(Q_{yt}) = R(Q_{yt}) - C(Q_{yt})    
\end{equation}
Notice that even though the good is produced at time $t$, it might get sold at any time in the future, moreover the price in the future is not known a priori to $y$. It is assumed that $y$ believes that it the whole amount, $Q_{yt}$ that will be sold at time $t+1$ and also, exhibiting naive expectations, $y$ believes that $P_{y(t+1)}=P_{yt}$. It also assumes that its demand for inputs will be met fully in the current time period. The demand for input from provider $x_i$ by $y$ at time $t$ is given as $q_{yx_it}$. Then,
\begin{equation}
  \pi(q_{yx_1t},q_{yx_2t}, \dots , q_{yx_\eta t},L_{yt}) = P_{yt}A_y L_{yt}^\beta \prod_{x_i \in \rho(y)}q_{yx_it}^{\alpha_i} - \sum_{x_i \in \rho(y)}P_{x_it}q_{yx_it} - \omega L_{yt}  
\end{equation}

Let $W_{yt}$ be the wealth of the producer $y$ at $t$. Then the budget constraint is given as
\begin{equation}
  C(Q_{yt}) = \sum_{x_i \in \rho(y)}P_{x_it}q_{yx_it} + \omega L_{yt} \leq W_{yt}  
\end{equation}
$W_{yt}$ is updated by the rule
\begin{equation}
    W_{y(t+1)} = W_{yt} + \pi_{yt}
\end{equation}

\subsection{The Consumer and the Utility Function}
For any particular consumer $y$, having utility elasticities of consumption as $\alpha_i$ for the particular good produced by the provider $x_i$, $\beta_i$ is the utility elasticity of leisure and $\gamma$ is the utility elasticity of income earned at time $t$. $q_{yx_it}$ is the quantity demanded by $y$ at time $t$ from the provider $x_i$. $V$ is the profit income of the consumer $y$. $T$ is the total available time for $y$ at any time period, to be divided between labour and leisure. Hence if $L_{yt}^s$ is the labour offered by $y$ at time $t$, then leisure $= T - L_{yt}^s $ if all the labour its offering is used up in some production. 

Wealth of $y$ at any time $t$ is given as $W_{yt}$ and is updated by the rule

\begin{equation}
    W_{y(t+1)}= W_{yt} + \omega_t L_{yt} + V_{yt}
\end{equation}  

Note that $\omega L_{yt} + V_{yt}$ is the income earned by $y$ at $t$ where . It must be mentioned again that as was in the case of producers, consumers have naive expectations about their profit income and expect all their demands to be completely fulfilled. 

\begin{equation}
    U(q_{yx_1t}, q_{yx_2t}, \dots , q_{yx_\eta t}, L_{yt}^s) = (\omega_t L_{yt}^s + V_{yt})^\beta (T - L_{yt}^s)^\gamma \prod_{x_i \in \rho(y)}q_{yx_it}^{\alpha_i}
\end{equation}

Note that $V_{yt}$ is unknown at the stage where the consumer is finding the optimal bundle, hence his best guess for $V_{yt} = V_{y(t-1)}$. The optimal bundle is obtained by maximising the above function subject to the budget constraint

\begin{equation}
  \sum_{x_i \in \rho(y)}P_{x_it}q_{yx_it} \leq W_{yt}  
\end{equation}

\subsection{The Optimization Problem}
The optimization problem for each agent is analytically solved and the solution is fed into the model as a “black-box” for the computer. In other words, the computer is not doing the optimization through some numerically constrained global optimization algorithm but rather the solution is already provided as a hard-coded rule. 

The aim of the following paragraphs are to arrive at these optimal maximizers and also to conclude upon the nature of a critical point without checking the signs of the principal leading minors of the hessian matrix explicitly. These reduce the computational load of the program allowing us to have a more efficient algorithm.

 Before the solution is presented the following proposition needs to be proved. 
\begin{theorem}
     The function 
    \[f:\mathbb{R}^{n+ } \mapsto \mathbb{R}^+ , \newline
    f(\textbf{x}) = PA\prod_{i=1}^{n} x_i^{\alpha_i} \] 
    is explicitly quasiconcave. P,A are positive constants and $0\leq \alpha_i \leq 1$.
\end{theorem}
\begin{proof}[Proof (by contradiction)]
    $\textbf{Case I }: f(\textbf{y}) >0, f(\textbf{x})=0 \\ \newline
    f(\textbf{x})=0 \Leftrightarrow \exists i, x_i=0\\
    \text{And} f(\textbf{y})> 0 \Leftrightarrow \forall i, x_i>0 \\ \newline      
    \text{Consider } w \doteq \theta \textbf{x} + (1-\theta)\textbf{y}, \theta \in (0,1] \\
    \Leftrightarrow \forall i, w_i = \theta x_i + (1-\theta)y_i > 0 \\
    \Leftrightarrow f(\textbf{w})>0 \\
    \Rightarrow f(\textbf{w})>f(\textbf{x}) \\
    \newline
    \textbf{Case II}: f(\textbf{y}) >0, f(\textbf{x})>0 \\ \newline
    \text{Let us assume that } f(\textbf{y}) > f(\textbf{x})  \text{ and } \nabla f(\textbf{x})\cdot (\textbf{y}- \textbf{x})= 0 \\
    \text{Note that f is quasiconcave}. \\ \newline
    \text{Consider} \nabla f(\textbf{y}) \cdot \nabla f(\textbf{x}) \\ 
    = f(\textbf{y}) f(\textbf{x}) 
    \begin{bmatrix} \vdots \\ 
    \frac{\alpha_i}{y_i} \\
    \vdots \\
    \end{bmatrix} 
    \cdot 
    \begin{bmatrix}\vdots \\ 
    \frac{\alpha_i}{x_i} \\
    \vdots \end{bmatrix} 
    = f(\textbf{y}) f(\textbf{x}) \sum_{i=1}^n \frac{\alpha_i^2}{x_i y_i} > 0 \\
    \newline \newline
    \text{Now let} \textbf{v} \in \mathbb{R}^{n+}
    \textbf{v} \doteq  \textbf{y} - \lambda \nabla f(\textbf{y}) \\
    \text{such that} f(\textbf{v})=f(\textbf{x}), \lambda > 0  \\ \newline
    \text{Consider} \nabla f(\textbf{x}) \cdot ( \textbf{v} - \textbf{x})
    =\nabla f(\textbf{x}) \cdot ( \textbf{v} - \textbf{y}) + \nabla f(\textbf{x}) \cdot ( \textbf{y} - \textbf{x}) \\
    = - \lambda \nabla f(\textbf{y}) \cdot \nabla f(\textbf{x}) <0 \\ \newline
    \text{However by quasiconcavity of f}, \nabla f(\textbf{x}) \cdot ( \textbf{v} - \textbf{x}) \geq 0 \\ 
    \text{Hence, a contradiction}. $
\end{proof}
   
It is known that when $\sum_{i=1}^{n}\alpha_i \leq 1$, $f$ is a continuous differentiable concave function. Hence, if it has a stationary point, it is the global maximum. The global maximum for any continuous differentiable function must be either a stationary point or a point in its boundary. So if $f$ doesn’t have a stationary point in the constraint set, then the global maximum must lie on the boundary. 
  
For $\sum_{i=1}^{n}\alpha_i > 1$ however, a stationary point is not necessarily even a local maximum. So the global maximum must lie at the boundary. 

We will first work with the case $\sum_{i=1}^{n}\alpha_i \leq 1$. Consider the function $f:\mathbb{R}^{n+ } \mapsto \mathbb{R} ,f(\textbf{x}) = PA\prod_{i=1}^{n} x_i^{\alpha_i} - \sum_{i=1}^n P_i x_i, P_i > 0$. This function is also concave as it is the linear combination of two concave functions. The solution to the unconstrained maximisation problem,  

\begin{equation}
    \max_{x_1,x_2,\dots,x_i} PA\prod_{i=1}^{n} x_i^{\alpha_i} - \sum_{i=1}^n P_i x_i
\end{equation}
is given by the matrix equation

\begin{equation}
\begin{bmatrix}(\alpha_1 -1) & \dots & \alpha_i & \dots & \alpha_n \\
\vdots & \ddots & \vdots & \vdots & \vdots \\
\alpha_1 & \dots & (\alpha_i-1) & \dots & \alpha_n \\
\vdots & \vdots & \vdots & \ddots & \vdots \\
\alpha_1 & \dots & \alpha_i & \dots & (\alpha_n -1) \end{bmatrix} \begin{bmatrix}log(x_1) \\
\vdots \\
log(x_i) \\
\vdots \\
log(x_n) \end{bmatrix} = \begin{bmatrix}log(\frac{P_1}{\alpha_1PA}) \\
\vdots \\
log(\frac{P_i}{\alpha_iPA}) \\
\vdots \\
log(\frac{P_n}{\alpha_nPA}) \end{bmatrix} 
\end{equation}

For the constrained case, the maximisation problem  

\begin{equation}
\begin{aligned}\max_{x_1,x_2,\dots,x_i} PA\prod_{i=1}^{n} x_i^{\alpha_i} - \sum_{i=1}^n P_i x_i \\
\textrm {s.t.} \sum_{i=1}^n P_i x_i = W
\end{aligned}
\end{equation}

is the same as
 
\begin{equation}
\begin{aligned}  \max_{x_1,x_2,\dots,x_i} PA\prod_{i=1}^{n} x_i^{\alpha_i}\\
\textrm {s.t.} \sum_{i=1}^n P_i x_i = W  \label{refeq} 
\end{aligned}
\end{equation}

The solution for which is 
\begin{equation}
    x_i^* = \frac{\alpha_i W}{P_i\sum_{i=1}^n \alpha_i}
\end{equation}
			
For the case of $\sum_{i=1}^{n}\alpha_i > 1$ however, no interior local maxima exists, and therefore the global maximum must exist at the boundary. We use the known result that for an explicitly quasiconcave function as an objective function and a convex set as the constrained domain, a critical point attained by constrained maximisation is a maxima and in this case that is the global maximum. \cite{chiang}

Observe that for any producer, its maximisation problem is exactly of the form we have been discussing. Hence, the following algorithm finds out the global maximum for any producer. 

If the sum of elasticities is less than or equal to one:
Find an unconstrained global maximising bundle 
If the cost of the bundle is less than the budget available:
Optimal bundle found 
Else if the cost of the bundle is more than the budget available:
Find the constrained maximising bundle. Optimal bundle found.
Else if sum of elasticities is more than one:
Find the constrained maximising bundle. Optimal bundle found.

In the case of consumers, the maximisation problem does not have any interior critical point, hence the maximum is found at the boundary. Hence the optimal bundle is found through constrained optimization, i.e, by exhausting the budget. 

Note that the optimization problem is the same as \ref{refeq} with $x_{n+1}, x_{n+2}, \alpha_{n+1}, \alpha_{n+2} = (\omega L_{yt} + V_{yt}), (T - L_{yt}),\beta, \gamma$ respectively. So the objective function remains of similar form, a new constraint equation however is added on top of the budget constraint equation.  
\begin{equation}
  x_{n+1} + \omega x_{n+2} = \omega_t T + V_{yt}  
\end{equation}

Hence the constraint equation is still a hyperplane, formed by the intersection of a hyperplane and a plane, and is a convex set. 
 
Therefore the consumer's optimal choice of their utility maximising bundle is simply the critical point of the constrained optimisation problem, and is given by the following. 
\begin{equation}
 q_{yx_it}^* = \frac{\alpha_i W_t}{P_i(\sum_{i=1}^n \alpha_i + \beta + \gamma)} 
 \end{equation}
 \begin{equation}
{L_{yt}^s}^* = \frac{\omega_t \beta T - \gamma V_t}{\omega_t(\beta + \gamma)}
\end{equation}
 
\subsection{Price and Wage Rate Adjustment}
Each period after trade has taken place in both the goods market and the labour market, producers produce and add it to their inventory stock. For any particular producer $y$, the demand faced at time $t$ is $Q_{yt}^d$. The price for period $ t+1 $ is set in the period $t$ and it is set by the adjustment rule:
\begin{equation}
 P_{y(t+1)} = P_{yt} + \sigma_y (Q_{y(t+1)}^d - Q_{y(t+1)}^s)   
\end{equation}	

Where $\sigma_y > 0 $ is the price adjustment factor or the rate of adjustment. However at time $t$, $Q_{y(t+1)}^d$ is not known, so the best guess for the producer is $Q_{yt}^d$. So the above equation becomes:
\begin{equation}
 P_{y(t+1)} = P_{yt} + \sigma_y (Q_{yt}^d - Q_{y(t+1)}^s)   
\end{equation}

If however the adjustment causes the price to be non-positive, then $\sigma_y$ is reduced by 10\% and then readjusted with the new value of $\sigma_y$, and the process is repeated until $P_{y(t+1)}>0$

By almost the exact same way, the wage rate is adjusted with two differences. Firstly, the wage rate being a global parameter, is adjusted by the economy. Secondly the labour supplied at time $t+1$ is not known at time $t$. Hence the wage rate is updated by the following rule. 
\begin{equation}
 \omega_{y(t+1)} = \omega_{yt} +  \Sigma(L_{t}^d - L_{t}^s)   
\end{equation}
Where $\Sigma$ is the wage adjustment factor.

\subsection{Detection of Equilibrium}
The prices for all the goods and services are constant over time iff the supply and demand are equal for all markets. We say that the economy has reached a general equilibrium. Consider the two cases, one where the marginal change of a dynamic variables with time is zero and the other when the \textit{net} marginal change of a dynamic variable with time is zero, we term both these cases as equilibrium. In other words, in all such cases where the dynamic variables show periodic predictable properties with its mean value having zero marginal change with time, we have considered them as equilibrium cases.

To this end, we have checked whether all the sequences of 100-period rolling average of prices converge to zero. Rolling average is taken to smooth out oscillations and a tolerance of $10^{-3}$ is taken for marginal changes. In the algorithm described below, the sequence ${p_{it}}$ is obtained by first taking the 100 period rolling average of the actual prices, then taking the absolute value of its first difference.

The value of $\epsilon$ in the following algorithm is $10^{-3}$. As the sequences are finite, it is ensured that at least the last 100 consecutive elements of the sequence are below the tolerance limit. Leaving these alterations necessary for a computational treatment, the rest of the procedure described is the formal definition of convergence in an algorithmic form.

The algorithm used to check for equilibrium in a model which ran for 1000 time periods is as follows. 

\begin{algorithm}
\caption{Algorithm for detecting equilibrium}\label{algo2}
\begin{algorithmic}[1]

\State Initialize an empty list $l$
\For{\textbf{each} $i$ \textbf{in} $producers$}
  \For{$t$ \textbf{in} $[500, 900)$}
    \If{$\forall x \in \{p_{it},p_{i(t+1)},\dots,p_{i1000}\}(x < \epsilon)$} 
            \State $l$\textbf{.push}(True)
            \State \textbf{break} 
        \EndIf
  \EndFor
  \State $l$\textbf{.push}(False) 
\EndFor
\If{$\forall \text{y} \in l$\text{(y==True)}}
    \State \textbf{return} $"Equilibrium"$
\Else 
    \State \textbf{return} $"Disequilibrium"$
\EndIf
 
\end{algorithmic}
\end{algorithm}

It must also be mentioned that the wage rate sequence is checked for equilibrium in a similar manner.

\section{Results}
In this section, we tabulate and illustrate the results obtained from running the ABM thus constructed. 10,000 simulations are run using parallel computing to obtain the results. The parameters are either chosen exactly or drawn randomly from a distribution as follows.

\begin{table}[h]
\caption{Exogenously given initial conditions}\label{tab1}%
\begin{tabular}{@{}ll|ll@{}}
\toprule
Parameter Name & Value/Distribution & Parameter Name & Value/Distribution\\
\midrule
No. of producers($P_n$)    & 10   & Initial producer wealth  & 1,000,000  \\
No. of consumers($C_n$)    & 80   & Initial consumer wealth  & 1,000  \\
Wage rate adjustment factor & 0.0005 & Price adjustment factor & 0.3 \\
Rate of technological change & 10 & Time period & 365 \\
Profit reinvestment ratio & 0.9 & Initial wage rate & 30 \\
Initial prices & Uniform(0,100) & Return to scale of producers & $|$Normal(0.9,0.6)$|$ \\
No. of shareholders of a firm \footnotemark[1]& Uniform(1,$C_n$) & Share in a firm(if shareholder)\footnotemark[2] & Uniform(0,1) \\
C-D function elasticities\footnotemark[3] & Uniform(0,1) & No. of providers\footnotemark[4] & Uniform(1,$P_n$) \\
\botrule
\end{tabular}
\footnotetext[1]{Shareholders sampled uniformly from the list of consumers.}
\footnotetext[2]{Renormalized such that sum of the shares of all the stakeholders of a company is 1.}
\footnotetext[3]{C-D: Cobb-Douglas. Renormalized such that sum of the elasticities is the return to scale.}
\footnotetext[4]{Providers sampled uniformly from the list of producers.}
\end{table}

Each of the individual simulation is run with the above parameters as exogenously given conditions. The values are sampled from the distribution and the set of all such sampled variables for a particular simulation is generated by the seeds $s_1$ and $s_2$. $s_1$ generates initial prices, return to scale of producers, no. of shareholders of a firm, share in a firm(if shareholders) and the shareholders of a firm. $s_2$ generates the Cobb-Douglas function elasticities, number of providers and the providers. In other words, $s_2$ generates the graph properties of the network economy and $s_1$ generates the rest of the initital conditions. 

The model is run for 1000 time periods if it is not halted by a termination condition. There are two termination conditions. One is when there remains only a single producer in the economy and the other is when the consumer wealth falls to zero. In both the cases, economic activities as envisioned in this model fails to take place or becomes meaningless in such conditions. Hence the model terminates in these cases. 

All the possible outcomes of a model run are therefore either that it reaches one of the termination conditions or it runs for 1000 time periods ultimately reaching equilibrium or staying in a perpetual disequilibrium. The results are tabulated in table \ref{tab2}. 

{\setlength\tabcolsep{3.5pt}
\begin{sidewaystable}
\caption{Outcomes of the simulation run}\label{tab2}
\begin{tabular*}{\textheight}{@{\extracolsep\fill}l  *{9}{p{1cm}}}
\toprule
Halt Type &  Count &  Total Wealth (Producer) &  Total Wealth (Consumer) &  Number of Shut Firms &  Gini Coefficient (Consumers) &  Gini Coefficient (Producers) &  Total Utility &    Wage Rate &  Leisure Proportion \\

\midrule
\textbf{ConsumerWealthZero} &    663 &              1.0080e+07 &              0.0000e+00 &            7.4510e+00 &                   7.7219e-01 &                   2.8095e-01 &     2.7666e+01 & 3.2220e+03 &          9.9827e+01 \\
    \textbf{Disequilibrium} &   7085 &              9.3838e+06 &              6.9625e+05 &            6.4837e+00 &                   5.2433e-01 &                   4.3397e-01 &     4.6654e+04 & 9.7409e+03 &          8.9598e+01 \\
       \textbf{Equilibrium} &    314 &              9.9130e+06 &              1.6704e+05 &            7.1338e+00 &                   3.7535e-01 &                   3.5480e-01 &     1.7370e+06 & 9.3168e+01 &          8.4041e+01 \\
\textbf{SingleProducerLeft} &   1938 &              8.6002e+06 &              1.4798e+06 &            7.9546e+00 &                   4.8176e-01 &                   3.6912e-01 &     1.2345e+08 & 2.8244e+04 &          9.4687e+01 \\

\botrule
\end{tabular*}
\footnotetext{Note: Except for count the other variables are the }
\end{sidewaystable}
}

We see that only 314 cases lead to equilibrium out of 10,000. Mostly the model stays in perpetual disequilibrium or halts due to a terminal condition. In all the four cases we see around 6 to 7 firms getting shut on average. In all the cases there exists inequality among both producers and consumers. This is noteworthy as we begin from a perfectly equal wealth endowment among all producers and all consumers. We also observe that equilibrium cases has the lowest average leisure proportion among all the other cases. It is also noteworthy that total utility is higher in equilibrium cases than the disequilibrium ones, however is the highest in the single producer case. 
Some illustrative examples are given for each of equilibrium and disequilibrium cases. The cases of single producer and consumer wealth going to zero has not been illustrated as the authors deemed these cases redundant. Curious readers however can inspect these cases themselves through the materials linked in the appendix.

\subsection{Equilibrium}

In this subsection we elaborate upon some of the typical equilibrium run cases observed. The first case to be illustrated has been generated by $s_1 = 78$ and $s_2 = 178$. 

The initial graph of the network economy is given as follows. 
\begin{figure}[H]
	\centering
	\includegraphics[width=\textwidth]{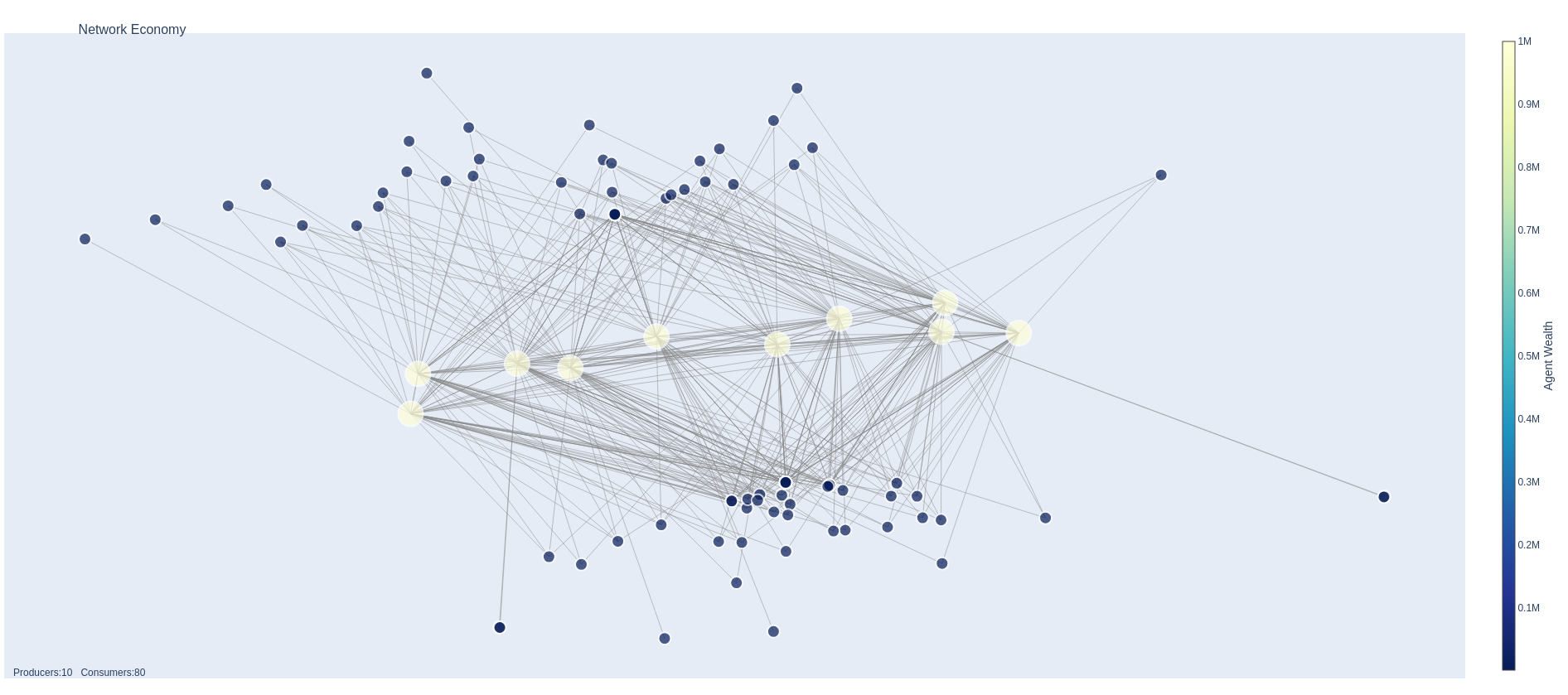}
	\caption{Model state with $s_1=78, s_2=178$ at $t = 0$}
	\label{fig00}
\end{figure}

This is the economy visualised as a network graph with the larger nodes representing the producers and the smaller ones the consumer. The colours of nodes denote the amount of wealth they have. The legend on the right helps get an idea of the value. 

After 1000 time periods only two producers remain. The resultant network graph is given in figure \ref{fig01}

\begin{figure}[H]
	\centering
	\includegraphics[width=\textwidth]{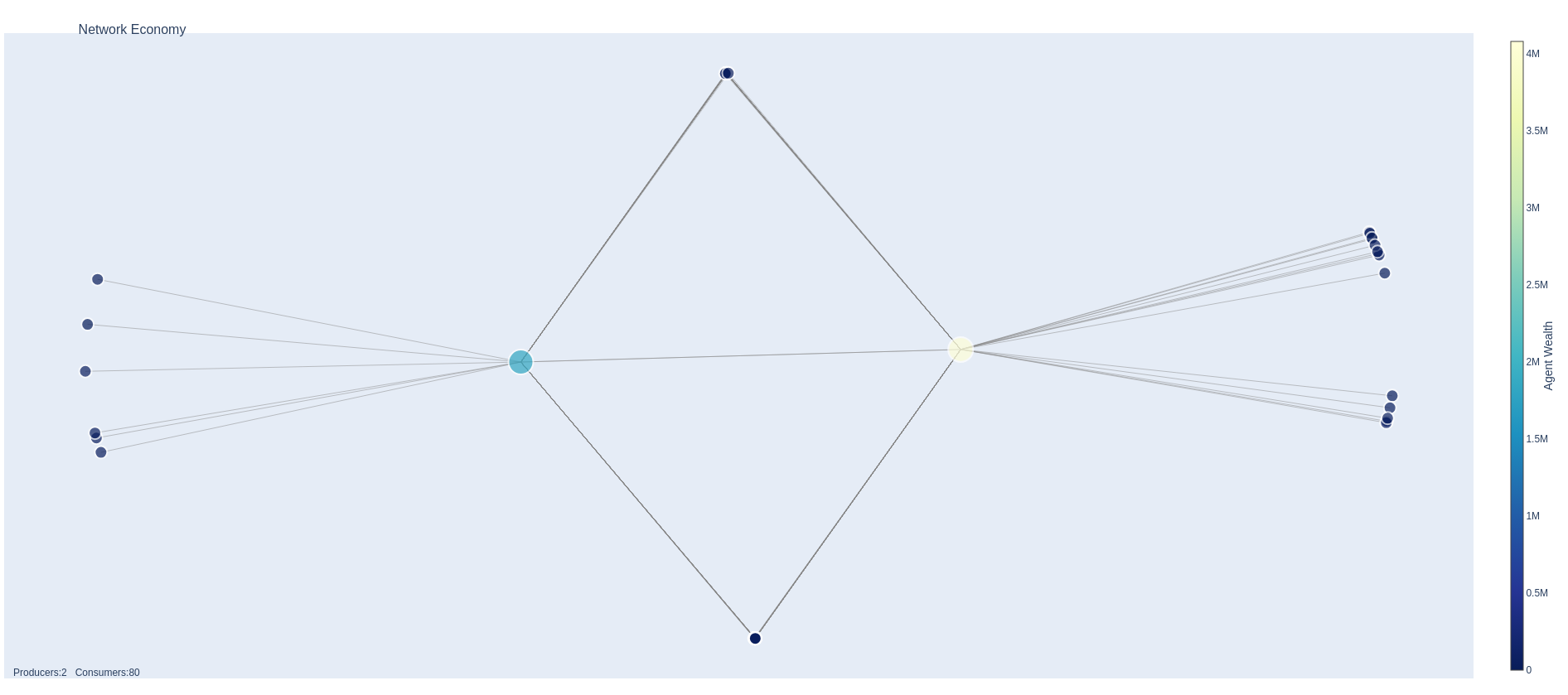}
	\caption{Model state at with $s_1=78, s_2=178$ $t = 1000$}
	\label{fig01}
\end{figure}

The evolution of various economic variables are illustrated in figure \ref{big image 1}.

\begin{figure}[H]
\begin{subfigure}{0.5\textwidth}
    \includegraphics[width=\linewidth]{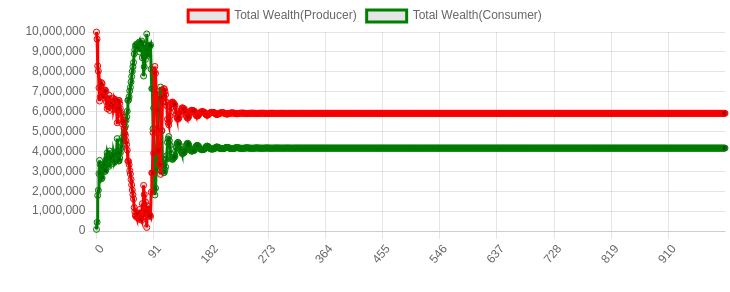}
    \caption{Evolution of wealth}
    \label{fig1}
\end{subfigure}
\begin{subfigure}{0.5\textwidth}
	\includegraphics[width=\linewidth]{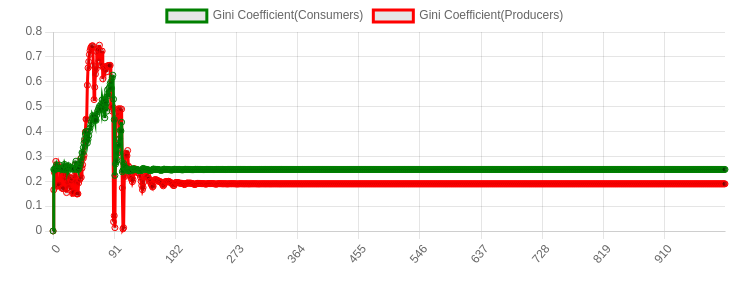}
	\caption{Evolution of gini coefficient}
	\label{fig1.5}
\end{subfigure}

\begin{subfigure}{0.5\textwidth}
    \includegraphics[width=\linewidth]{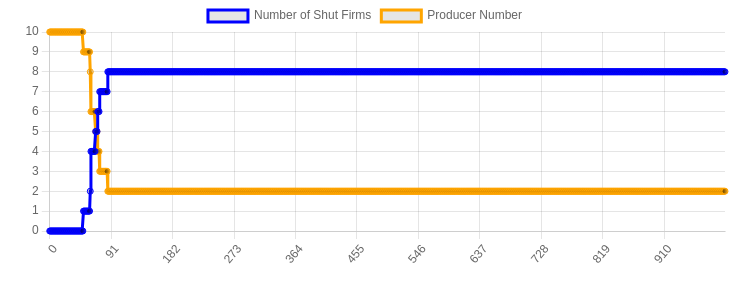}
    \caption{No. of producers/ shut firms}
    \label{fig2}
\end{subfigure}
\begin{subfigure}{0.5\textwidth}
    \includegraphics[width=\linewidth]{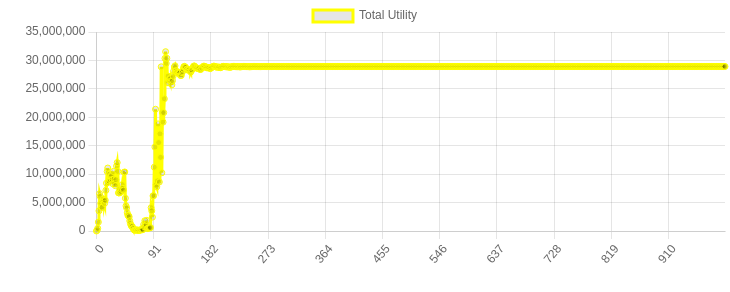}
    \caption{Evolution of total utility}
    \label{fig3}
\end{subfigure}
\begin{subfigure}{0.5\textwidth}
 	\includegraphics[width=\linewidth]{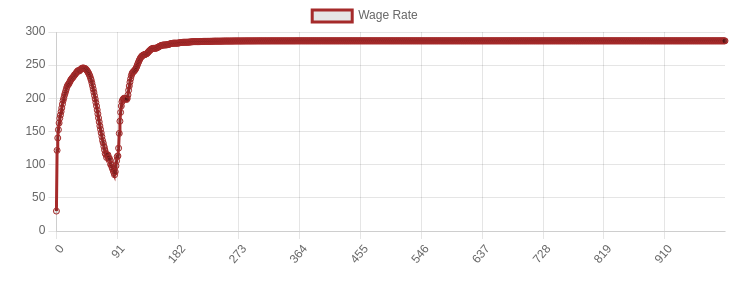}
	\caption{Evolution of total utility}
	\label{fig4}   
\end{subfigure}
\begin{subfigure}{0.5\textwidth}
	\includegraphics[width=\linewidth]{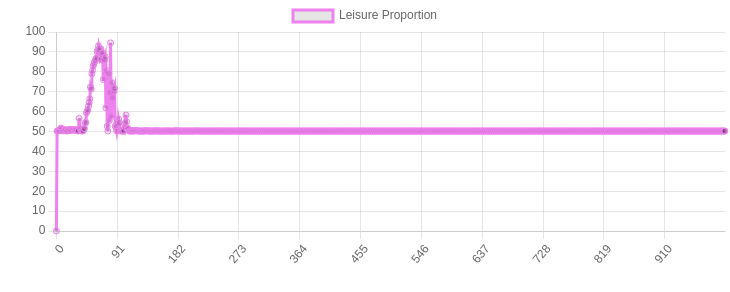}
	\caption{Evolution of aggregate leisure proportion}
	\label{fig5}    
\end{subfigure}
\begin{subfigure}{0.5\textwidth}
        \includegraphics[width=\linewidth]{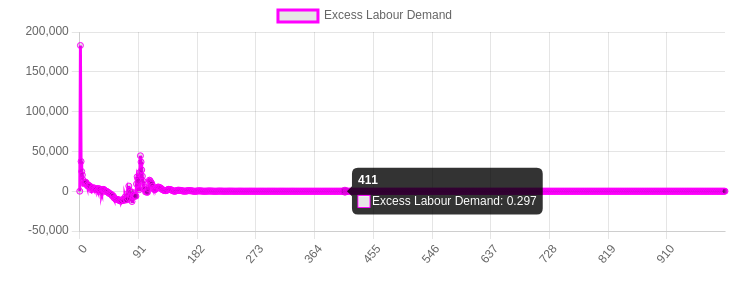}
	\caption{Evolution of excess labour demand}
	\label{fig6}
\end{subfigure}
\caption{Model with $s_1=78, s_2=178$}
\label{big image 1}
\end{figure}








We can see that all the economic variables converges to a certain value. Also note that the total utility and wage rate of the economy has gone up significantly from its initial value. Furthermore, leisure proportion is quite low, signifying high employment in the economy. The excess labour demand curve is also converging to zero hence there is no involuntary employment as well. 

It needs to be pointed out especially that the case just described is an ideal normatively desirable case. In fact in models with $s_1=30, s_2=11$ and $s_1=78,s_2=157$ we end up in equilibrium with 6 and 9 producers and while keeping all the desirable properties of the above one as well.

It must however be pointed out that not all equilibrium cases have these desirable properties. One such example is generated by $s_1=47,s_2=121$.

The network graphs of the economy is given in figure \ref{fig10} and \ref{fig11}. We can see that after 1000 time periods only two producers remain.

\begin{figure}[H]
	\centering
	\includegraphics[width=\textwidth]{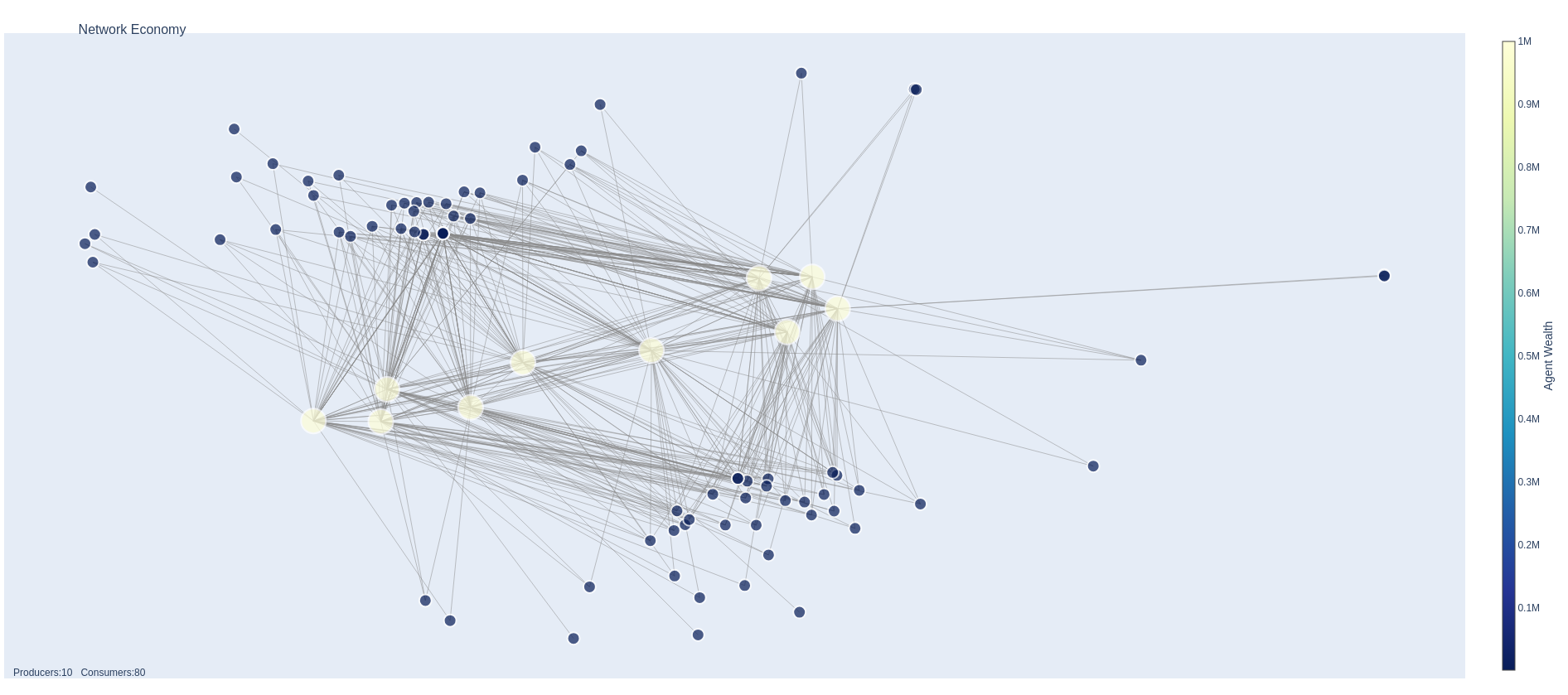}
	\caption{Model state with $s_1=47, s_2=121$ at $t = 0$}
	\label{fig10}
\end{figure}
\begin{figure}[H]
	\centering
	\includegraphics[width=\textwidth]{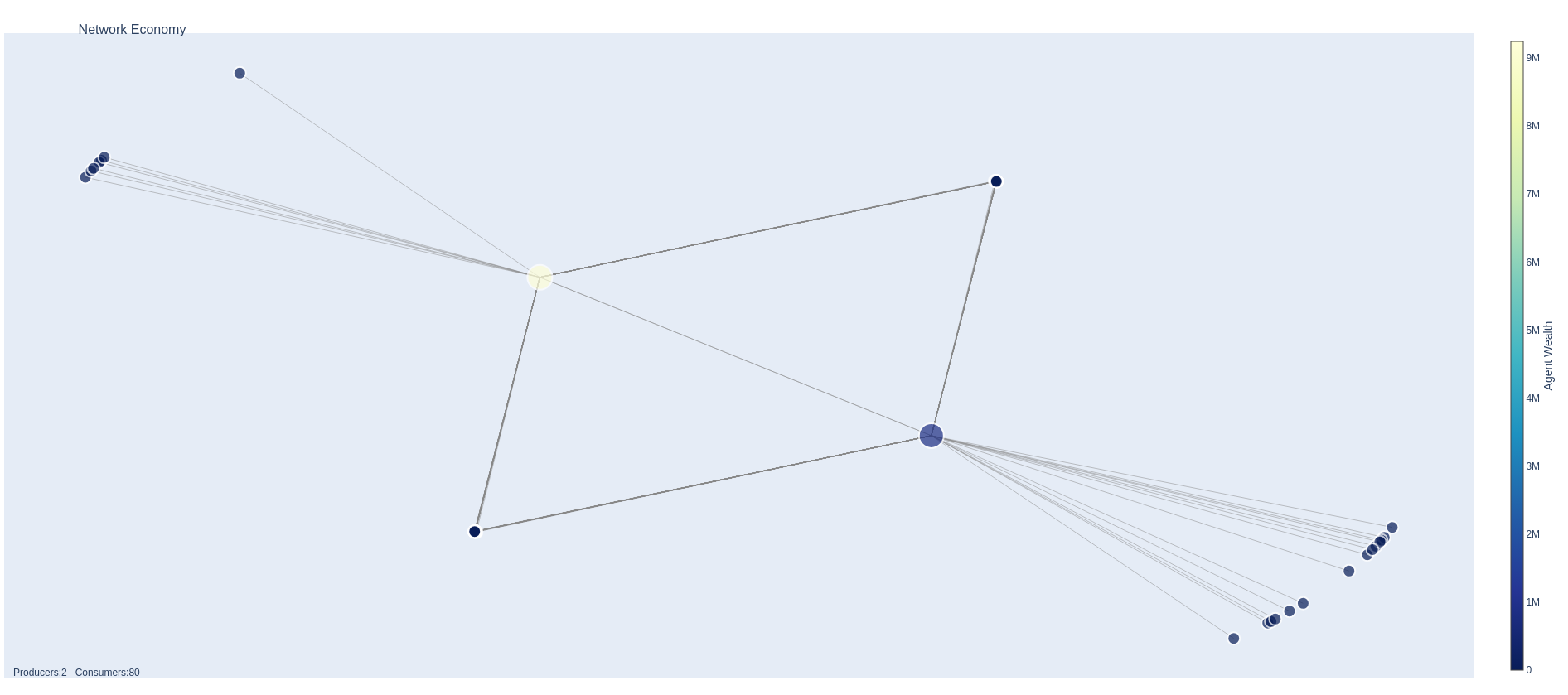}
	\caption{Model state with $s_1=47, s_2=121$ at $t = 1000$}
	\label{fig11}
\end{figure}

The economic variables evolve as in figure \ref{big image 2}

\begin{figure}[H]
\begin{subfigure}{0.5\textwidth}
    \includegraphics[width=\linewidth]{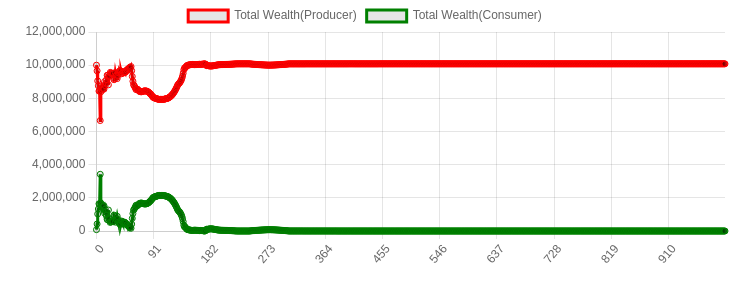}
    \caption{Evolution of wealth}
    \label{fig1.1}
\end{subfigure}
\begin{subfigure}{0.5\textwidth}
	\includegraphics[width=\linewidth]{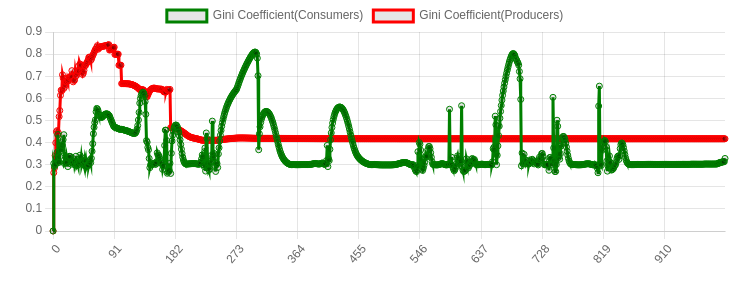}
	\caption{Evolution of gini coefficient}
	\label{fig1.2}
\end{subfigure}
\begin{subfigure}{0.5\textwidth}
    \includegraphics[width=\linewidth]{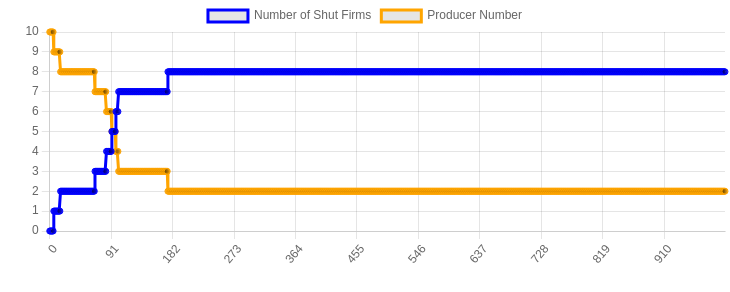}
    \caption{No. of producers/ shut firms}
    \label{fig1.3}
\end{subfigure}
\begin{subfigure}{0.5\textwidth}
    \includegraphics[width=\linewidth]{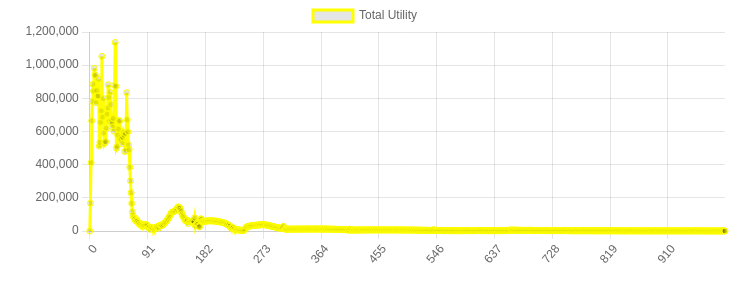}
    \caption{Evolution of total utility}
    \label{fig1.4}
\end{subfigure}
\begin{subfigure}{0.5\textwidth}
 	\includegraphics[width=\linewidth]{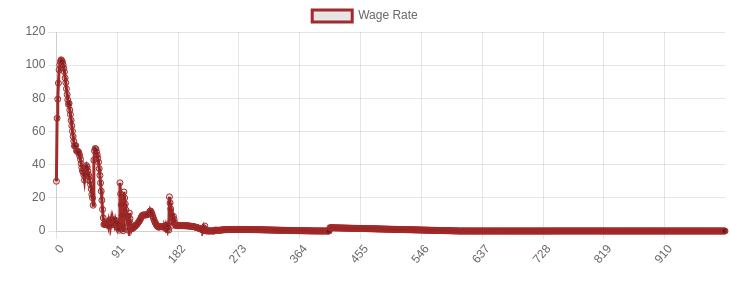}
	\caption{Evolution of total utility}
	\label{fig1.5}   
\end{subfigure}
\begin{subfigure}{0.5\textwidth}
	\includegraphics[width=\linewidth]{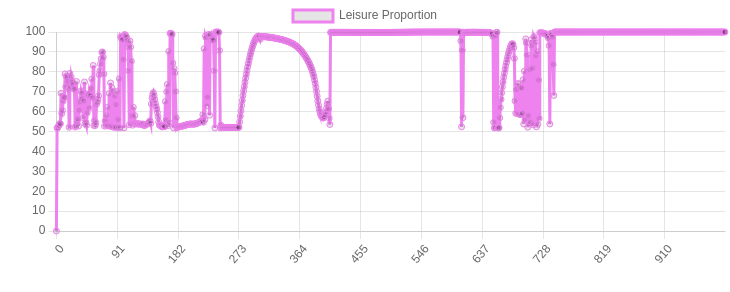}
	\caption{Evolution of aggregate leisure proportion}
	\label{fig1.6}    
\end{subfigure}
\begin{subfigure}{0.5\textwidth}
        \includegraphics[width=\linewidth]{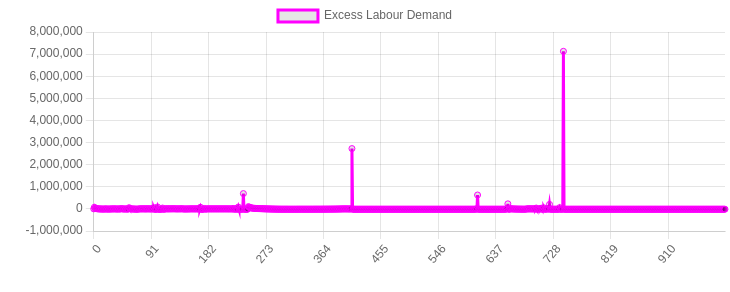}
	\caption{Evolution of excess labour demand}
	\label{fig1.7}
\end{subfigure}
\begin{subfigure}{0.5\textwidth}
        \includegraphics[width=\linewidth]{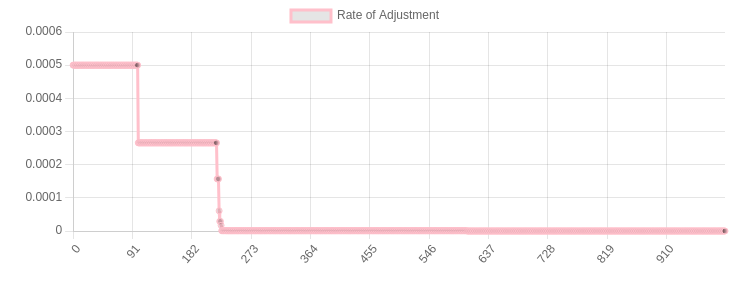}
	\caption{Evolution of rate of adjustment of wage}
	\label{fig1.8}
\end{subfigure}
\caption{Model with $s_1=47, s_2=121$}
\label{big image 2}
\end{figure}

Notice how in this case the total utility level has floored. Notice also that the wage rate and the utility function seem to move together. At time period  1000 aggregate leisure proportion is 99.999999 and the aggregate excess labour demand is 1.406763e+04 implying that there is heavy unemployment. However, this simulation run passes the convergence test as the rate of adjustment of wage has fallen almost to zero and for this reason the changes in the wage rate is almost negligible and similarly for all other prices.

This is made possible by the fact that even with non-zero amount of excess demand the price(wage) adjustment factor can fall to zero. The price adjustment factor is reduced to 0.9 of its previous value every time it tries to adjust to a negative price. As we see in figure \ref{fig1.7}, the high spikes in the excess labour demand forces the wage rate adjustment factor to plummet to near zero levels. However there is no mechanism for the adjustment factor to increase. By design, it is a weakly decreasing function of time. This is done to ensure that the adjustment factor does not over-correct the prices near possible equilibrium values. 

The modeler has to therefore ensure explicitly the excess demand is also satisfactorily low for a run marked as equilibrium to get the aforementioned desirable properties.

\subsection{Disequilibrium}

In this subsection we illustrate a typical model in disequilibrium. The seeds of the example model will be $s_1=25,s_2=112$. The model ends up with 3 producers at the end of 1000 time periods as shown by figure \ref{fig20} and \ref{fig21}.

\begin{figure}[h!]
	\centering
	\includegraphics[width=\textwidth]{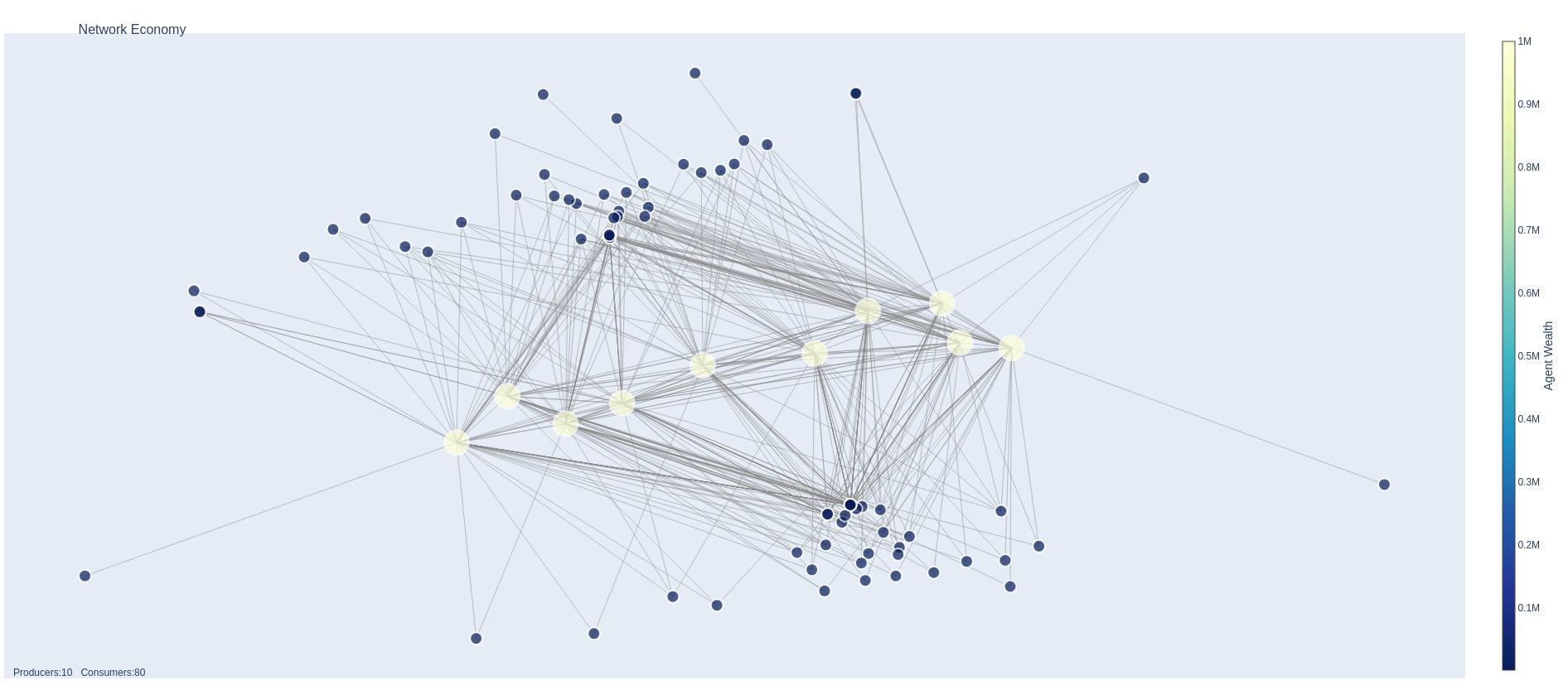}
	\caption{Model state with $s_1=25, s_2=112$ at $t = 0$}
	\label{fig20}
\end{figure}
\begin{figure}[h!]
	\centering
	\includegraphics[width=\textwidth]{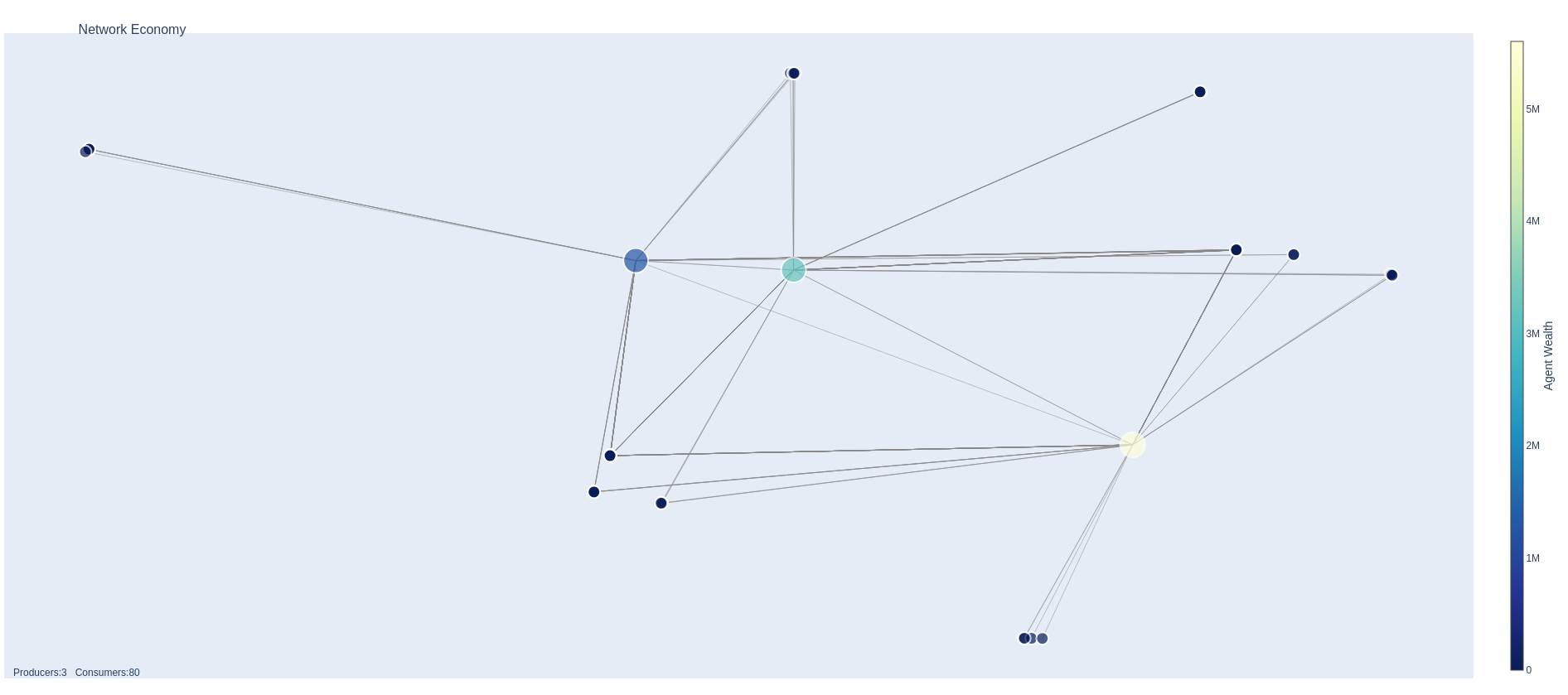}
	\caption{Model state with $s_1=47, s_2=112$ at $t = 1000$}
	\label{fig21}
\end{figure}

The evolution of the economic variables are shown in figure \ref{big image 3}.
\begin{figure}[H]
\begin{subfigure}{0.5\textwidth}
    \includegraphics[width=\linewidth]{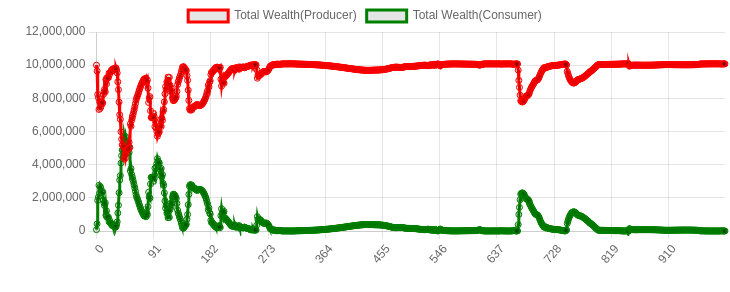}
    \caption{Evolution of wealth}
    \label{fig2.1}
\end{subfigure}
\begin{subfigure}{0.5\textwidth}
	\includegraphics[width=\linewidth]{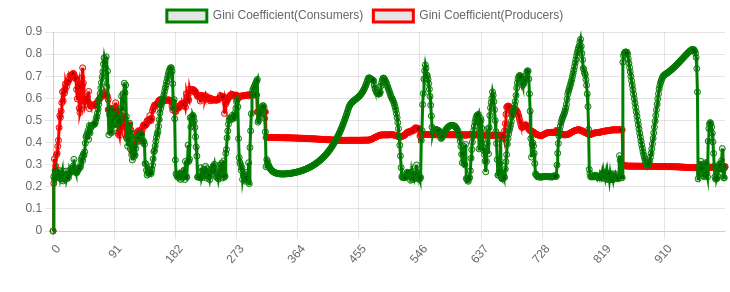}
	\caption{Evolution of gini coefficient}
	\label{fig2.2}
\end{subfigure}
\begin{subfigure}{0.5\textwidth}
    \includegraphics[width=\linewidth]{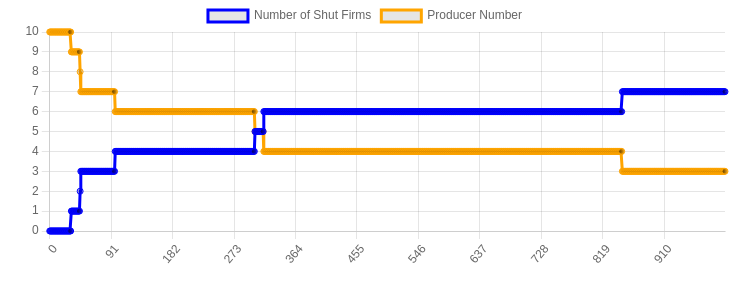}
    \caption{No. of producers/ shut firms}
    \label{fig2.3}
\end{subfigure}
\begin{subfigure}{0.5\textwidth}
    \includegraphics[width=\linewidth]{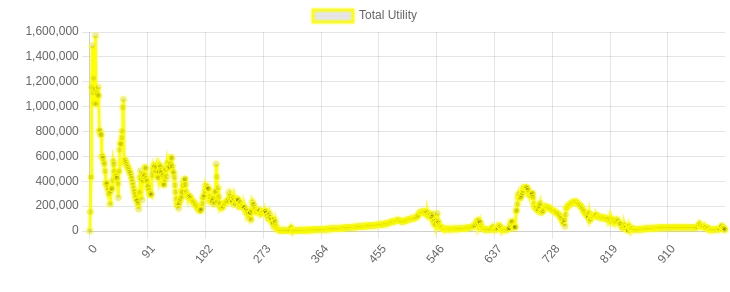}
    \caption{Evolution of total utility}
    \label{fig2.4}
\end{subfigure}
\begin{subfigure}{0.5\textwidth}
 	\includegraphics[width=\linewidth]{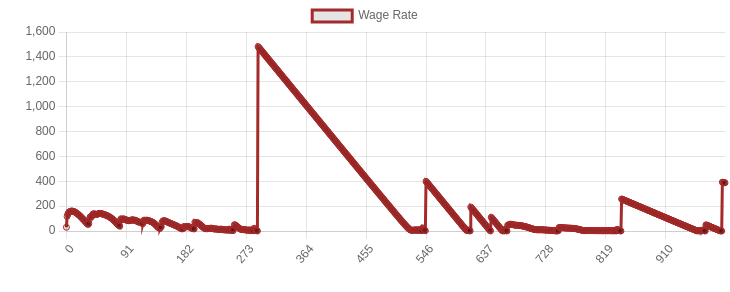}
	\caption{Evolution of total utility}
	\label{fig2.5}   
\end{subfigure}
\begin{subfigure}{0.5\textwidth}
	\includegraphics[width=\linewidth]{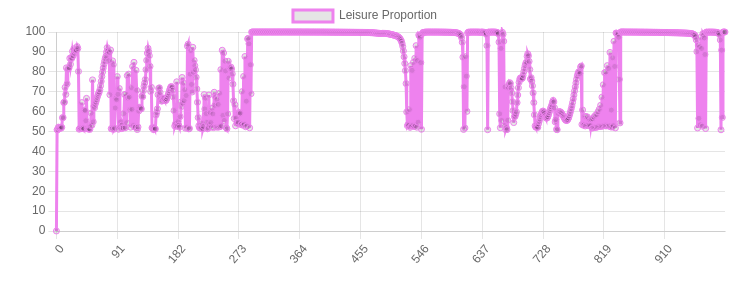}
	\caption{Evolution of aggregate leisure proportion}
	\label{fig2.6}    
\end{subfigure}
\begin{subfigure}{0.5\textwidth}
        \includegraphics[width=\linewidth]{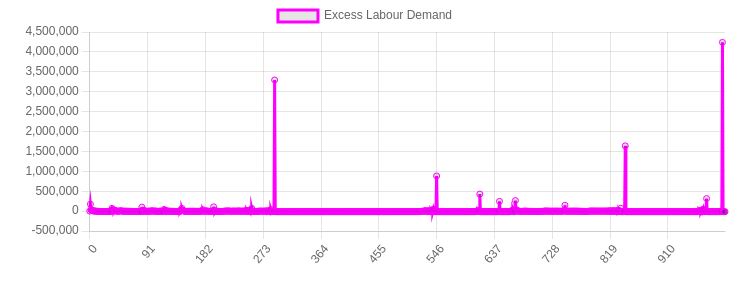}
	\caption{Evolution of excess labour demand}
	\label{fig2.7}
\end{subfigure}
\begin{subfigure}{0.5\textwidth}
        \includegraphics[width=\linewidth]{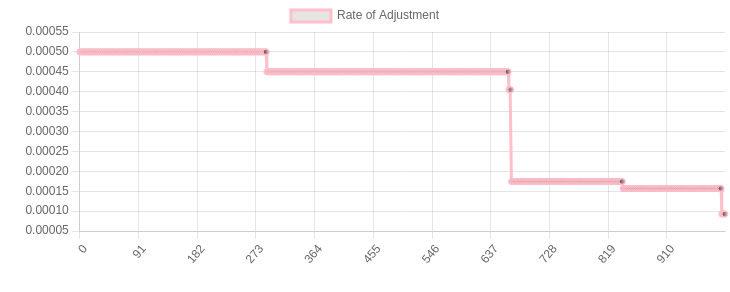}
	\caption{Evolution of rate of adjustment of wage}
	\label{fig2.8}
\end{subfigure}
\caption{Model with $s_1=25, s_2=112$}
\label{big image 3}
\end{figure}

We see that the economy is in perpetual disequilibrium. The wage rate curve is characterized by steep rises and then gradual reverting back to near zero levels. We also see that phases of high wage rates are accompanied by near 100 leisure proportion. This is because when the wage rates get too high, firms move to less labour intensive methods and hence there is no adequate demand for labour. It is also noteworthy that the utility curve keeps decreasing from its initial values.

The above simulation run paints rather a grim picture for the cases of disequilibrium. However, one interesting case, generated by seeds $s_1=78, s_2=122$, yields a situation where the economy is in disequilibrium yet has all the desirable properties of model seed $s_1=78, s_2=178$.

The model ends up with 2 producers as shown in figure 

\begin{figure}[h!]
	\centering
	\includegraphics[width=\textwidth]{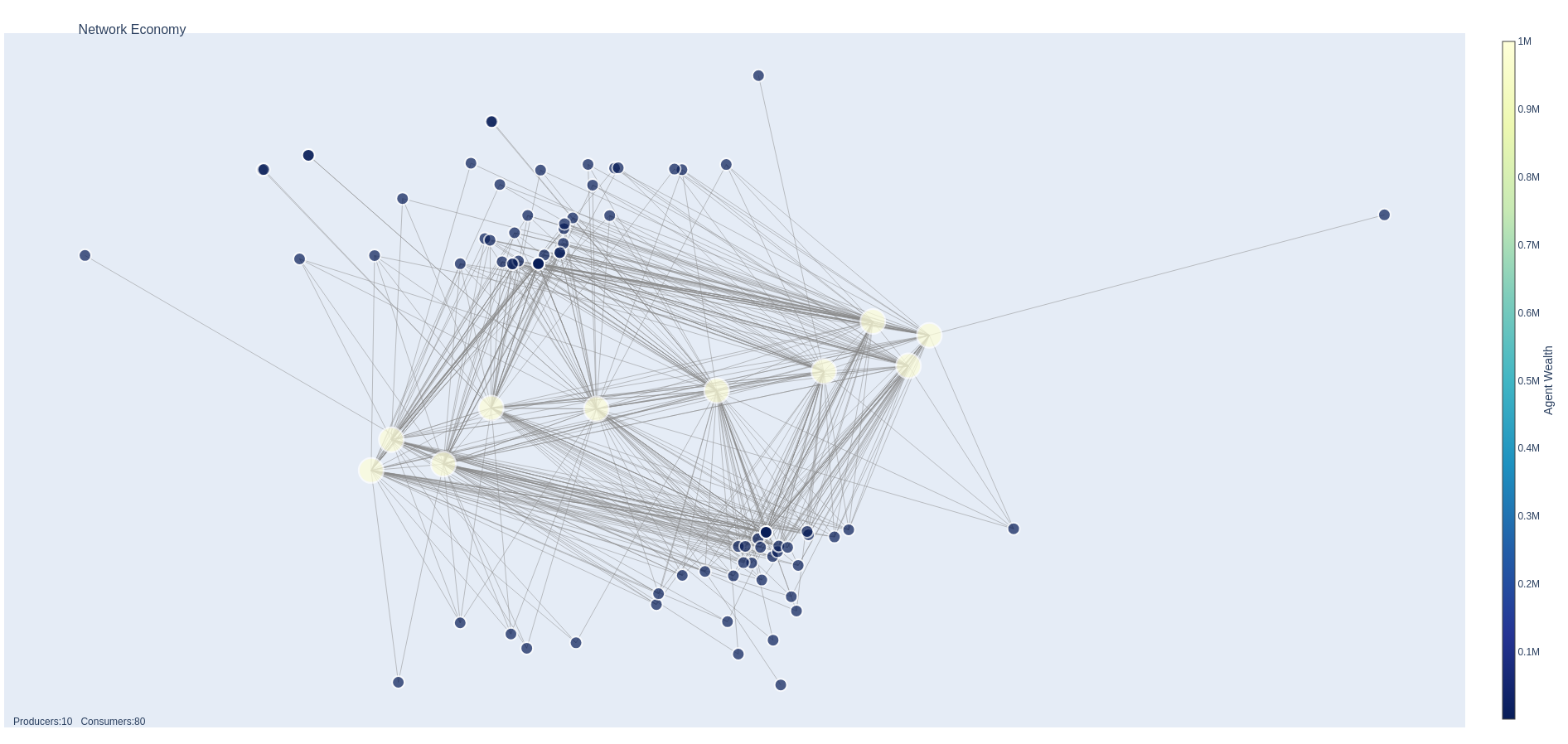}
	\caption{Model state with $s_1=78, s_2=122$ at $t = 0$}
	\label{fig30}
\end{figure}
\begin{figure}[h!]
	\centering
	\includegraphics[width=\textwidth]{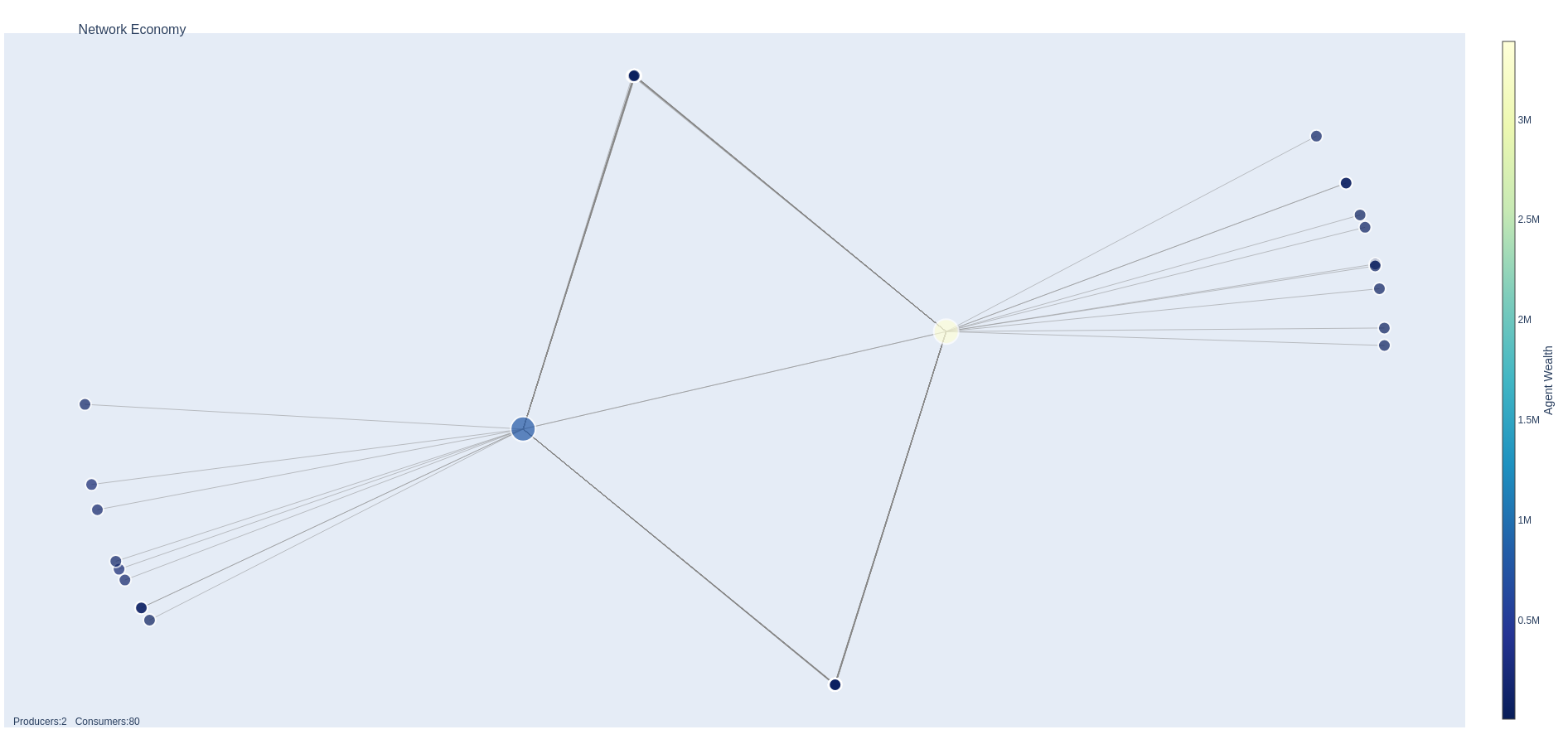}
	\caption{Model state with $s_1=78, s_2=122$ at $t = 1000$}
	\label{fig31}
\end{figure}

The evolution of the economic variables are shown in figure \ref{big image 4}.
\begin{figure}[H]
\begin{subfigure}{0.5\textwidth}
    \includegraphics[width=\linewidth]{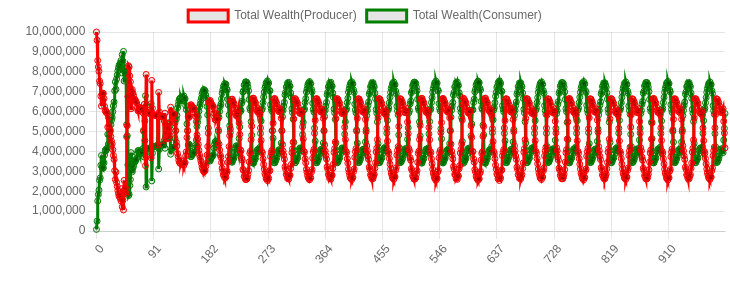}
    \caption{Evolution of wealth}
    \label{fig3.1}
\end{subfigure}
\begin{subfigure}{0.5\textwidth}
	\includegraphics[width=\linewidth]{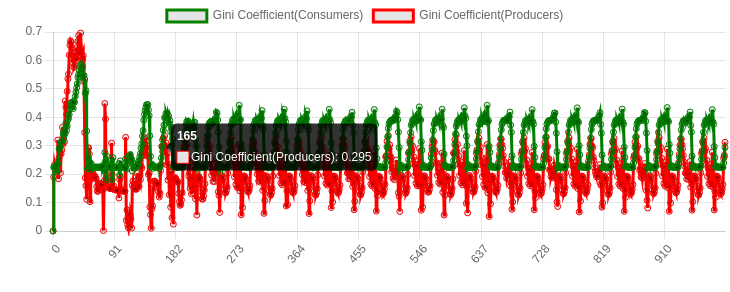}
	\caption{Evolution of gini coefficient}
	\label{fig3.2}
\end{subfigure}
\begin{subfigure}{0.5\textwidth}
    \includegraphics[width=\linewidth]{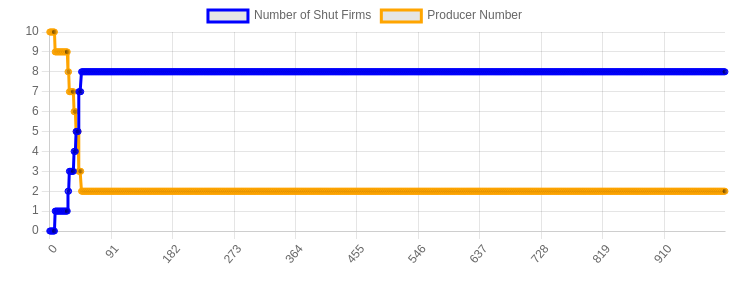}
    \caption{No. of producers/ shut firms}
    \label{fig3.3}
\end{subfigure}
\begin{subfigure}{0.5\textwidth}
    \includegraphics[width=\linewidth]{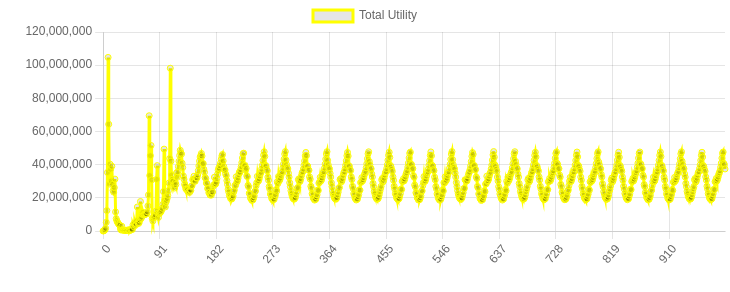}
    \caption{Evolution of total utility}
    \label{fig3.4}
\end{subfigure}
\begin{subfigure}{0.5\textwidth}
 	\includegraphics[width=\linewidth]{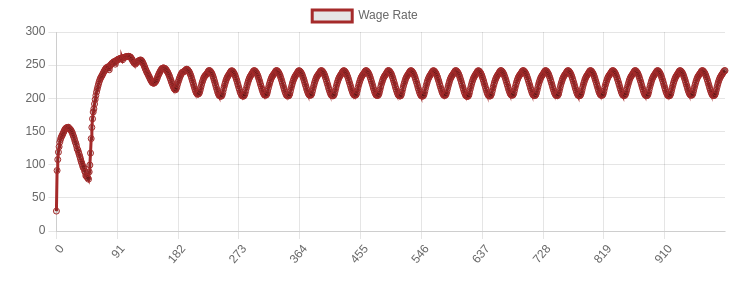}
	\caption{Evolution of total utility}
	\label{fig3.5}   
\end{subfigure}
\begin{subfigure}{0.5\textwidth}
	\includegraphics[width=\linewidth]{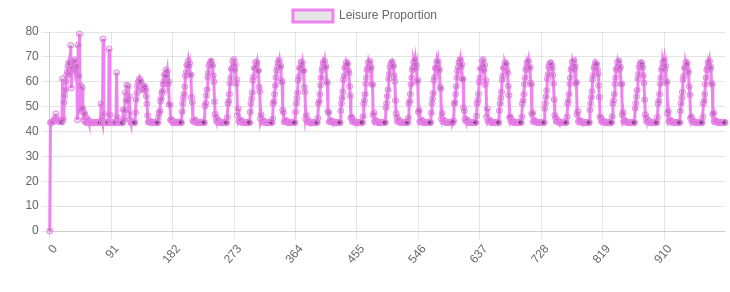}
	\caption{Evolution of aggregate leisure proportion}
	\label{fig3.6}    
\end{subfigure}
\begin{subfigure}{0.5\textwidth}
        \includegraphics[width=\linewidth]{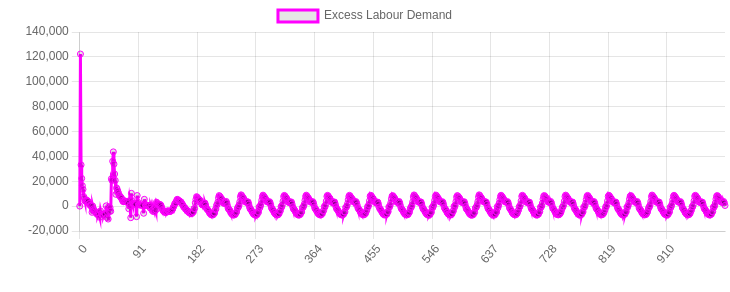}
	\caption{Evolution of excess labour demand}
	\label{fig3.7}
\end{subfigure}
\begin{subfigure}{0.5\textwidth}
        \includegraphics[width=\linewidth]{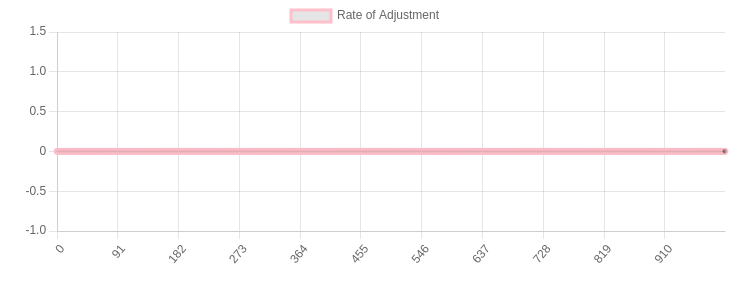}
	\caption{Evolution of rate of adjustment of wage}
	\label{fig3.8}
\end{subfigure}
\caption{Model with $s_1=78, s_2=122$}
\label{big image 4}
\end{figure}

We see that all the economic variables are oscillating with approximately constant amplitude and constant frequency. The utility level has gone up from the initial vales and even though its oscillating the leisure proportion is noticeably low. Excess labour demand is oscillating around zero level. 

We get to see all the desirable macroeconomic properties and the prices(wage) seem to have come to a predictable oscillating pattern. However, the amplitude of these oscillations are too big for the rolling average transformation to smooth them into a linear constant function. Hence these are classified as an equilibrium case.

\section{Conclusion}

We have created an agent-based model of the real economy with production and trade networks, relaxing some of the assumptions of the neoclassical theory, like perfect competition, capital aggregation, non-increasing returns to scale and homogeneous products and agents. 

The resultant model is a general-purpose open-source model which can be used for studying emergent patterns and performing counterfactual analysis.

10,000 simulations with 10 producers and 80 consumers were run with the model created and the following observations have been made. 

\begin{itemize}
    \item 314 runs of 10,000 lead to an equilibrium condition. This suggests that only small pockets of appropriate initial conditions can lead to an equilibrium condition.
    \item Most of the simulations stayed in a perpetual disequilibrium. However, some of the simulations resulted in consumer wealth falling to zero or only a single producer remaining in the whole economy. This suggests that from certain initial conditions a purely decentralized market structure with rational optimizing agents can also lead to economic disasters.
    \item In all the cases, starting from an equal endowment of wealth, an arbitrary assignment of trade and production relations, and elasticities leads to inequality in society.
    \item We see that in both equilibrium and disequilibrium, on average 6 to 7 firms were shut down.
    \item In some of the equilibrium cases, there existed high total utility, wage rate and no involuntary employment. These cases showed desirable normative properties a well functioning economy would agreeably be characterised by. 
    \item In some cases marked as equilibrium however, there existed non-zero excess demand, however the rates of adjustment were so low in these simulation cases that the prices and wage rate each converged to constant values. The aforementioned desirable properties were generally not present in these cases.
    \item In some cases marked as disequilibrium, the economy reached a sort of stable state where the economic parameters exhibited periodic oscillating behaviours with near-constant amplitude and frequency. However the amplitudes were too big to smooth the curves into a linear constant function by a rolling average transformation. Hence these cases were marked as disequilibrium. However, these cases also showed to some extent the desirable properties of some of the equilibrium cases.
\end{itemize}

\bmhead{Supplementary information}

The source code, accompanying files and the data generated by this model is open source and is available in a public git repository linked in the appendix. Economic researchers and developers are welcome to send in pull requests for fixing potential bugs, features and refactors. 

\section{Discussion}
In this section the authors discuss the scope for further development of this work. In terms of the modelling the economy, there exists a few features which can be included in future versions. These are listed as follows. 

\begin{itemize}
    \item The price(wage) adjustment rate can be modified so that it can try to offset for persistent non zero excess demand. Cyclic rate of adjustment is a possible alternative. \cite{smith2017cyclical}
    \item Agents can be modified to allow for a factor for labour efficiency, which would then make the labour output of each agent heterogeneous.
    \item Asset markets can be introduced.
    \item Institutions like government and central banks can be introduced.
    \item Finite or infinite horizon forward looking utility functions can be introduced to induce the agents to save their earnings.
    \item New entrants into the market can be allowed for and the population can be made to grow(or shrink).
    \item The functional form used is of Cobb-Douglas type and other functional forms can be introduced into the model.
    \item Instead of naive expectations, firms can have some sort of adaptive expectations about future parameters. Also instead of rational optimization, artificial intelligence algorithms can guide their decision making, to make it more "human-like".
\end{itemize}
  
The results obtained from these simulation runs can be further analysed. One immediate further development would be to find a characterization of the initial conditions which lead to equilibrium and another would be to find a characterization of those firms which remain and those which close down.

In terms of performance, the authors acknowledge that Python may not be the best programming language choice in future developments. However, the MESA module has been chosen for its ease of use and code readability. Agents.jl module in Julia can be an option for performance enhancement.\cite{Agents.jl}
 
\section*{Statements and Declarations}

\begin{itemize}
\item Funding: The authors declare that no funds, grants, or other support were received during the preparation of this manuscript.
\item Conflict of interest/Competing interests: The authors have no relevant financial or non-financial interests to disclose. 
\item Ethics approval: Not applicable 
\item Consent to participate: Not applicable
\item Consent for publication: Not applicable
\item Availability of data and materials: Available on Github.
\item Code availability: Available on Github. Link : \url{https://github.com/subhamonsey/ABMEconomy}
\item Authors' contributions: This is an extended work based on Subhamon Supantha's masters dissertation under the guidance of Prof. Naresh Kumar Sharma. All the main ideas and concepts have been formulated, debated and discussed by both the authors equally. The algorithms, proofs and the code has been developed by Subhamon Supantha. The first draft has been written by Subhamon Supantha and all authors commented on previous versions of the manuscript. All authors read and approved the final manuscript.
\end{itemize}
\section*{Appendix}\label{secA1}

Link to the data, supplementary files and source code: \url{https://github.com/subhamonsey/ABMEconomy}

\bibliography{sn-bibliography} 


\begin{thebibliography}{10}
\ifx \bisbn   \undefined \def \bisbn  #1{ISBN #1}\fi
\ifx \binits  \undefined \def \binits#1{#1}\fi
\ifx \bauthor  \undefined \def \bauthor#1{#1}\fi
\ifx \batitle  \undefined \def \batitle#1{#1}\fi
\ifx \bjtitle  \undefined \def \bjtitle#1{#1}\fi
\ifx \bvolume  \undefined \def \bvolume#1{\textbf{#1}}\fi
\ifx \byear  \undefined \def \byear#1{#1}\fi
\ifx \bissue  \undefined \def \bissue#1{#1}\fi
\ifx \bfpage  \undefined \def \bfpage#1{#1}\fi
\ifx \blpage  \undefined \def \blpage #1{#1}\fi
\ifx \burl  \undefined \def \burl#1{\textsf{#1}}\fi
\ifx \doiurl  \undefined \def \doiurl#1{\url{https://doi.org/#1}}\fi
\ifx \betal  \undefined \def \betal{\textit{et al.}}\fi
\ifx \binstitute  \undefined \def \binstitute#1{#1}\fi
\ifx \binstitutionaled  \undefined \def \binstitutionaled#1{#1}\fi
\ifx \bctitle  \undefined \def \bctitle#1{#1}\fi
\ifx \beditor  \undefined \def \beditor#1{#1}\fi
\ifx \bpublisher  \undefined \def \bpublisher#1{#1}\fi
\ifx \bbtitle  \undefined \def \bbtitle#1{#1}\fi
\ifx \bedition  \undefined \def \bedition#1{#1}\fi
\ifx \bseriesno  \undefined \def \bseriesno#1{#1}\fi
\ifx \blocation  \undefined \def \blocation#1{#1}\fi
\ifx \bsertitle  \undefined \def \bsertitle#1{#1}\fi
\ifx \bsnm \undefined \def \bsnm#1{#1}\fi
\ifx \bsuffix \undefined \def \bsuffix#1{#1}\fi
\ifx \bparticle \undefined \def \bparticle#1{#1}\fi
\ifx \barticle \undefined \def \barticle#1{#1}\fi
\bibcommenthead
\ifx \bconfdate \undefined \def \bconfdate #1{#1}\fi
\ifx \botherref \undefined \def \botherref #1{#1}\fi
\ifx \url \undefined \def \url#1{\textsf{#1}}\fi
\ifx \bchapter \undefined \def \bchapter#1{#1}\fi
\ifx \bbook \undefined \def \bbook#1{#1}\fi
\ifx \bcomment \undefined \def \bcomment#1{#1}\fi
\ifx \oauthor \undefined \def \oauthor#1{#1}\fi
\ifx \citeauthoryear \undefined \def \citeauthoryear#1{#1}\fi
\ifx \endbibitem  \undefined \def \endbibitem {}\fi
\ifx \bconflocation  \undefined \def \bconflocation#1{#1}\fi
\ifx \arxivurl  \undefined \def \arxivurl#1{\textsf{#1}}\fi
\csname PreBibitemsHook\endcsname

\bibitem[\protect\citeauthoryear{Arrow and Debreu}{1954}]{arrowdebreu}
\begin{barticle}
\bauthor{\bsnm{Arrow}, \binits{K.J.}},
\bauthor{\bsnm{Debreu}, \binits{G.}}:
\batitle{Existence of an equilibrium for a competitive economy}.
\bjtitle{Econometrica}
\bvolume{22}(\bissue{3}),
\bfpage{265}--\blpage{290}
(\byear{1954}).
Accessed 2022-12-14
\end{barticle}
\endbibitem

\bibitem[\protect\citeauthoryear{Ackerman}{2002}]{ackerman}
\begin{barticle}
\bauthor{\bsnm{Ackerman}, \binits{F.}}:
\batitle{Still dead after all these years: interpreting the failure of general
  equilibrium theory}.
\bjtitle{Journal of Economic Methodology}
\bvolume{9},
\bfpage{119}--\blpage{139}
(\byear{2002})
\end{barticle}
\endbibitem

\bibitem[\protect\citeauthoryear{Robinson}{1953}]{Joan}
\begin{barticle}
\bauthor{\bsnm{Robinson}, \binits{J.}}:
\batitle{The production function and the theory of capital}.
\bjtitle{The Review of Economic Studies}
\bvolume{21}(\bissue{2}),
\bfpage{81}--\blpage{106}
(\byear{1953}).
Accessed 2023-07-11
\end{barticle}
\endbibitem

\bibitem[\protect\citeauthoryear{Evans and Honkapohja}{2001}]{Expectations}
\begin{bchapter}
\bauthor{\bsnm{Evans}, \binits{G.W.}},
\bauthor{\bsnm{Honkapohja}, \binits{S.}}:
\bctitle{Expectations, economics of}.
In: \beditor{\bsnm{Smelser}, \binits{N.J.}},
\beditor{\bsnm{Baltes}, \binits{P.B.}} (eds.)
\bbtitle{International Encyclopedia of the Social $\&$ Behavioral Sciences},
pp. \bfpage{5060}--\blpage{5067}.
\bpublisher{Pergamon},
\blocation{Oxford}
(\byear{2001}).
\doiurl{10.1016/B0-08-043076-7/02245-2} .
\burl{https://www.sciencedirect.com/science/article/pii/B0080430767022452}
\end{bchapter}
\endbibitem

\bibitem[\protect\citeauthoryear{Tesfatsion}{2006}]{tesfatsion2006agent}
\begin{barticle}
\bauthor{\bsnm{Tesfatsion}, \binits{L.}}:
\batitle{Agent-based computational economics: A constructive approach to
  economic theory}.
\bjtitle{Handbook of computational economics}
\bvolume{2},
\bfpage{831}--\blpage{880}
(\byear{2006})
\end{barticle}
\endbibitem

\bibitem[\protect\citeauthoryear{Arthur}{2006}]{arthur2006out}
\begin{barticle}
\bauthor{\bsnm{Arthur}, \binits{W.B.}}:
\batitle{Out-of-equilibrium economics and agent-based modeling}.
\bjtitle{Handbook of computational economics}
\bvolume{2},
\bfpage{1551}--\blpage{1564}
(\byear{2006})
\end{barticle}
\endbibitem

\bibitem[\protect\citeauthoryear{Masad and Kazil}{2015}]{mesa}
\begin{bchapter}
\bauthor{\bsnm{Masad}, \binits{D.}},
\bauthor{\bsnm{Kazil}, \binits{J.}}:
\bctitle{Mesa: an agent-based modeling framework}.
In: \bbtitle{14th Python in Science Conference},
vol. \bseriesno{2015},
pp. \bfpage{53}--\blpage{60}
(\byear{2015}).
\bcomment{Citeseer}
\end{bchapter}
\endbibitem

\bibitem[\protect\citeauthoryear{Chiang}{2005 - 2005}]{chiang}
\begin{bbook}
\bauthor{\bsnm{Chiang}, \binits{A.C.}}:
\bbtitle{Fundamental Methods of Mathematical Economics / Alpha C. Chiang, Kevin
  Wainwright.},
\bedition{Fourth edition.} edn.
\bpublisher{McGraw-Hill/Irwin},
\blocation{Boston, Mass}
(\byear{2005 - 2005})
\end{bbook}
\endbibitem

\bibitem[\protect\citeauthoryear{Smith}{2017}]{smith2017cyclical}
\begin{botherref}
\oauthor{\bsnm{Smith}, \binits{L.N.}}:
Cyclical Learning Rates for Training Neural Networks
(2017)
\end{botherref}
\endbibitem

\bibitem[\protect\citeauthoryear{Datseris et~al.}{2022}]{Agents.jl}
\begin{barticle}
\bauthor{\bsnm{Datseris}, \binits{G.}},
\bauthor{\bsnm{Vahdati}, \binits{A.R.}},
\bauthor{\bsnm{DuBois}, \binits{T.C.}}:
\batitle{Agents.jl: a performant and feature-full agent-based modeling software
  of minimal code complexity}.
\bjtitle{{SIMULATION}}
\bvolume{0}(\bissue{0}),
\bfpage{003754972110688}
(\byear{2022})
\doiurl{10.1177/00375497211068820}
\end{barticle}
\endbibitem

\end{thebibliography}
	

\end{document}